\documentclass[12pt]{article}

\usepackage{amssymb}
\usepackage{amsfonts}
\usepackage{amsmath}
\usepackage{setspace}
\usepackage{parskip}
\usepackage{geometry}
\usepackage{scalefnt}
\usepackage{xcolor}
\usepackage{bbm}
\usepackage{hyperref,cleveref}
\usepackage{enumerate}
\usepackage{natbib}
\usepackage[bottom]{footmisc}
\usepackage[normalem]{ulem}

\newtheorem{assumption}{Assumption}

\newtheorem{example}{Example}

\newtheorem{lemma}{Lemma}

\newtheorem{observation}{Observation}

\newtheorem{proposition}{Proposition}
\newtheorem{remark}{Remark}

\newcommand{\ind}{\mathbbm{1}}

\newcommand{\E}{\mathrm E}
\newcommand{\bsig}{\mathbf{\mathrm\Sigma}}
\newcommand{\hpr}{\mathrm H^{\,0}_{\mathrm{PRE}}}
\newcommand{\hpt}{\mathrm H^{\,0}_{\mathrm{PT}}}
\newcommand{\taus}{\tau_{it}^0}

\newenvironment{proof}[1][Proof]{\noindent\textbf{#1.} }{\hfill\rule{0.5em}{0.5em}}
\DeclareMathOperator*{\argmax}{argmax}
\DeclareMathOperator*{\sgn}{sgn}

\title{\Large Heterogeneous Treatment Effects via\\Dynamic Panel Data Methods
\footnote{We thank participants at the IAAE Conference in Turin, as well as seminar audiences at Emory, Georgetown, Johns Hopkins, Singapore Management University, U of British Columbia, and U of Georgia for their helpful comments. 
}} 

\author{Philip Marx\thanks{Louisiana State
    University. Email: \href{mailto:pmarx@lsu.edu}{pmarx@lsu.edu}.}
\and  Elie Tamer\thanks{Harvard University. Email:
    \href{mailto:elietamer@fas.harvard.edu}{elietamer@fas.harvard.edu}.}
\and Xun Tang\thanks{Rice University.
    Email: \href{mailto:xun.tang@rice.edu}{xun.tang@rice.edu}.}
}
\date{\today}

\begin{document}

\maketitle

\begin{abstract} 
We study dynamic panel data methods in a causal framework where treatment effects are heterogeneous and history-dependent.
Sequential exchangeability alone does not guarantee that a first-differenced instrumental-variable (IV) estimand is a convex average of current effects, even if those effects are constant. 
We characterize additional conditions that are sufficient and necessary for eliminating trend contamination and guaranteeing a convex representation uniformly over compatible distributions of potential outcomes and treatments. 
Next, for a model with autoregressive untreated potential outcomes, we propose a \emph{conditional} first-differenced IV regression that identifies the average treatment effect on the first-time treated \emph{even without} sequential exchangeability. 
Finally, we compare sequential exchangeability and parallel trends under sequential randomization and forward-looking treatment choices with learning, and show that, when treatments depend on recent outcomes, equal pre-treatment trends can coexist with violations of parallel trends.
\end{abstract}

\medskip %\vspace{0.25in}

\begin{quote}
\end{quote}

\newpage

\section{Introduction}

Linear dynamic panel estimators are commonly used in settings with persistent outcomes and unobserved unit-level heterogeneity \citep{balestranerlove66, andersonhsiao82, arellanobond91}. 
This includes settings where the goal is treatment effect estimation. 
For example, \cite{acemoglu2018democracy} study the effect of countries' democratization on their growth and motivate their dynamic panel regression by the need to control for the rich dynamics of GDP (the outcome) while specifying democratization (the treatment) as a reversible treatment process. 
Dynamic panel methods are also common in production function estimation \citep{blundellbond2000}, where they have been used more specifically to study whether exporting improves firm productivity, e.g., \cite{clerides1998export, biesenbrock2005, newman2017}.
In addition to the dynamics of productivity, this literature documents the importance not just of unit-level productivity heterogeneity \citep{melitz2003trade, syverson2011productivity}, but also heterogeneity in the potential productivity gains from exporting \citep{lileevatreffler2010}, i.e., the treatment effects.

A natural question is what causal effect, if any, is identified by the pooled treatment coefficient in such a dynamic regression when the underlying causal model may be heterogeneous.
We study this question for a first-differenced IV estimand motivated by the dynamic panel methods of \cite{andersonhsiao82} and \cite{arellanobond91}. 

To answer this question, we first derive a simple decomposition of the first-differenced IV treatment coefficient (Lemma \ref{thm:iv-fwl}) in terms of potential outcomes in a simple three-period setting, within a causal framework that allows those potential outcomes to depend flexibly on treatment histories. 
We then consider its interpretation under sequential exchangeability, a non-parametric assumption that current and future potential outcomes are (mean) independent from current treatment conditional on a full history of observed outcomes and treatments.
While sequential exchangeability (with suitable overlap) already identifies a variety of average potential outcomes and dynamic treatment effects,%
\footnote{
In particular, there exists a broad suite of estimators collectively referred to as $g$-methods [e.g., see \cite{robins1986g, robins1997msm} and \cite{robinshernan2009chapter} for a survey].
} 
its role in the first part of our analysis is to provide a benchmark under which adjustment for observed history identifies causal effects, allowing us to assess whether the IV coefficient aggregates such effects appropriately. 
% isolate the question of aggregation via a linear regression from confounding due to dynamic selection.
This exercise parallels the literature decomposing and interpreting two-way fixed effect regression coefficients assuming parallel trends \citep{DS2020,goodman2021difference} or two-stage least squares (2SLS) coefficients under the standard LATE assumptions \citep{blandhol2026tsls}, where the underlying assumptions already identify causal effects.%
\footnote{
The two-way fixed effect literature has also considered the role and the complications of covariates \citep{heckmanichimuratodd1997xdid, abadie2005xdid, caetanocallaway2023xdid}, as well as proposing alternative estimators under parallel trends and other assumptions on the outcome or treatment processes, like staggered designs (e.g., \cite{DS2020, callawaysatanna, sundid}).
Among existing works that build on variants of parallel trends, the setting of \cite{chaisemartinhaultfoeuille2024dyn} is most similar to ours in that it also allows for intertemporal effects on potential outcomes in non-staggered designs; in contrast, however, our approach is based on assumptions other than parallel trends.
}
In each case, the question is what causal interpretation a regression coefficient admits under assumptions on potential outcomes, without imposing a homogeneous linear outcome model per se.

We find that sequential exchangeability alone does not guarantee that the treatment coefficient is a convex average of current treatment effects. 
Additional complications arise in our setting because the lagged outcome regressor itself reflects earlier treatment effects and potential outcome trends. 
Heterogeneity in these earlier responses enters its first-stage prediction and thus the construction of the IV weights.
This can contaminate the coefficient with the untreated trend and prevent the treatment effect weights from summing to one.

Building on our decomposition, we characterize restrictions on heterogeneity that eliminate untreated-trend contamination and normalize the treatment-effect weights uniformly over compatible data-generating processes.
Failures of these restrictions can produce a negative coefficient even when current treatment effects are constant and strictly positive. 
We then characterize an additional restriction on period-2 treatment-effect heterogeneity that guarantees a squared-weight convex representation uniformly over permissible period-2 assignment rules. 
This restriction also characterizes whether the estimand is bounded across such rules: when the condition fails, the coefficient can become arbitrarily positive or negative for a fixed potential outcome distribution. 

The focus on sequential exchangeability in the preceding analysis naturally raises two follow-up questions: (a) Can causal inference be conducted using dynamic panel estimators without sequential exchangeability? (b) Can sequential exchangeability be justified in observational data where treatments are endogenously determined by forward-looking individuals?
We provide conditions under which the answers to both questions are indeed positive.
In the latter case, we also compare sequential exchangeability to the parallel trends assumption that is common in the econometrics literature.

To answer (a), \Cref{sec:causal-no-se} delineates cases where sequential exchangeability is generally violated, and yet existing or modified dynamic panel estimators identify well-defined causal estimands. 
As a benchmark, we show in \Cref{apx:ARPO} that, for a model where all potential outcomes follow a first-order auto-regression process with homogeneous treatment effects, the IV estimand identifies the constant effect under conditions that violate sequential exchangeability.

We then consider a more general model in which only the \textit{untreated} potential outcomes follow a first-order autoregressive process, while the intertemporal treatment effects are allowed to be time-varying and heterogeneous.
We show that in this case, an IV regression conditional on observed history identifies a conditional average treatment effect on the treated, as well as time trends and the state dependence parameter.
This result requires two assumptions: conditional serial uncorrelation of the structural errors of potential outcomes (\Cref{assn:arupo}), and a conditional mean independence restriction on treatment effects (\Cref{assn:MI-hte}).
The latter restricts only how heterogeneous treatment effects relate to the latent components of untreated outcomes, namely the fixed effects and past shocks. It does not rule out dependence of treatment choices on the observed history or on those latent components, and so it does not impose sequential exchangeability.

To investigate question (b), Section \ref{sec:evaluating} provides, in the context of both experimental and observational data, conditions under which sequential exchangeability holds whereas a common alternative, the parallel trends condition, does not. For an {\it experimental} setting, we show sequentially randomized treatments are compatible with sequential exchangeability, but not compatible with parallel trends in general.

We also provide guidance for empiricists to assess sequential exchangeability and parallel trends in {\it observational} data, especially when treatment choices depend on past outcomes. 
To this end, we consider a structural model of dynamic treatment choices by rational, forward-looking units, 
and use it to show the kinds of restrictions that are required, especially on the distribution of prior beliefs, to yield sequential exchangeability or parallel trends. 
With Gaussian learning, we provide a sharp characterization of conditions on the units' beliefs, fixed effects, and state dependence under which these two assumptions are justified. 
Thus, our work connects to recent work on identification of learning models in nonparametric settings (e.g., \cite{bunting2022learning}).   This exercise is similar to the one done in \cite{marx2024parallel} for a model under parallel trends.

\Cref{sec:pretrend} turns to the power of pre-trend tests, which are commonly used in practice to support the parallel trends assumption. \cite{ghanem2026when} investigate this under general treatment rules, including those dependent upon unobserved fixed effects of units. Our emphasis is instead on the loss of power given dynamics: untreated outcomes are persistent, and later treatment depends on past outcomes. We provide conditions under which the pre-trend null holds while the parallel trends condition fails. 
For example, when treatment depends on pre-treatment outcomes, units are sorted into treatment groups partly on the most recent transitory shock. This generates selection bias in untreated outcome trends over post-treatment periods, even though pre-trends are identical between the treated and untreated groups under the stationarity of pre-treatment outcomes and their symmetry in the treatment rule.

Our work also relates to a broader agenda relating sequential exchangeability and parallel trends (e.g., \cite{renson2023pt}) and the corresponding estimators, for example through ``bracketing'' relationships on the target parameter \citep{angrist2008mostly, ding2019bracketing}.
Additionally, \cite{han2021dynamic} studies the identification of dynamic treatment effects beyond sequential exchangeability, namely under a dynamic version of rank similarity. 
In complementary work, \cite{kim2023dynamic} studies identification of a dynamic panel data model when treatment is staggered and potential outcomes are non-negative, statically indexed, and take a multiplicative form; in this case, he obtains identification results for ratios of treatment effects under sequential exchangeability conditional on persistent heterogeneity. Finally, \cite{klosin} proposes an approach that corrects for the bias caused by ignoring dynamic feedback in static panel data regressions. 

The rest of the paper is organized as follows. \Cref{sec:framework} introduces the dynamic potential outcome framework. \Cref{sec:DPR} studies causal inference using dynamic panel data methods, with and without the sequential exchangeability condition. \Cref{sec:evaluating} investigates when sequential exchangeability or parallel trends can be justified in an \emph{observational} context of dynamic treatment choices with learning, and an \emph{experimental} setting of sequentially randomized treatments. \Cref{sec:pretrend} examines how dynamics in potential outcomes and treatment selection affect the power of pre-trend tests for the null hypothesis of parallel trends. \Cref{sec:conclusion} concludes. Proofs are collected in \Cref{apx:proofs} and \ref{apx:proofs2}, and extensions in \Cref{apx:supp-causal-no-se} and \ref{apx:supp-learning}.

\section{A Dynamic Framework for Potential Outcomes} \label{sec:framework}

The researcher observes binary treatment $D_{it}$ with realized outcomes $Y_{it}$ for individual units $i$ in time periods $t=0,1,\dots,T$. 
There is a pre-treatment period $t=0$ with an initial realized outcome $Y_{i0}$ and $\Pr\{D_{i0} = 0\}=1$.
Subsequent dates index panel periods, not each unit's first treatment date; treatment may begin, stop, or recur after baseline.
We do not model how $Y_{i0}$ is determined prior to the sampling period, but allow it to be correlated with potential outcomes and treatment choices in our analysis.
Potential outcomes in period $t\geq1$ are functions only of an individual's own treatment history through period $t$. 
\bigskip 

\begin{assumption}
\label{assn:po-treatments}
Potential outcomes for individual unit $i$ in period $t$ are indexed only by $i$'s own treatment through that time period, i.e., $Y_{it} (d_{i}^t)$, with $d^{t}_i\equiv (d_{is})_{s=1}^{t}$. 
All potential outcomes, including $Y_{i0}$, have finite second moments.
\end{assumption}
\bigskip

Assumption \ref{assn:po-treatments} rules out spillovers across individual units as well as anticipation of future treatments affecting present potential outcomes. 
Observed outcomes equal potential outcomes evaluated at observed treatments, $Y_{it} = Y_{it} (D_i^t)$.
In this framework, the treatment effects $Y_{it}(d^t) - Y_{it}(\widetilde d^t)$ are \textit{intertemporal} in that they are defined by the full history of past potential treatments.

We adopt notation $Y_i^t \equiv (Y_{i0}, Y_{i1}, \dots, Y_{it})$ for the history of realized outcomes, $Y_i^t (d^t) \equiv (Y_{i0}, Y_{i1} (d^1), \dots, Y_{it} (d^t))$ for that of potential random outcomes, and $Y_{i}^t (\cdot)\equiv \{Y_i^t(d^t):d^t\in\{0,1\}^t\}$, i.e., the collection of potential random vectors through period $t$.  
Throughout, we restrict to realized and potential treatment vectors $D_i^t$ and $d_i^t$ with degenerate initial realization $D_{i0} = d_{i0} = 0$, which we suppress where convenient.

One assumption we consider is a form of sequential exchangeability; it posits some notion of unconfoundedness conditional on the history of treatments and \emph{observed} outcomes. We drop individual subscripts $i$ for simplicity. 

\vspace{.25cm}

\begin{assumption}[Sequential Exchangeability]
\label{assn:mse}
Let $1 \leq s \leq t \leq T$, and let $d^{s}$ and $d^t$ denote subvectors of the first $s$ and $t$ terms in $d^T\equiv (d_1,\dots, d_{s}, \dots, d_t, \dots, d_T)$ respectively. 
\begin{enumerate}[a.]
    \item Under (full) sequential exchangeability,
    \[
        (\, Y_s (d^s), \dots, Y_T (d^T) \, ) \perp D_s \mid  Y^{s-1}, D^{s-1} = d^{s-1}.
    \]

    \item  Under (mean) sequential exchangeability,
    \[
        \E[ \, Y_t (d^t) \mid  Y^{s-1}, D^{s-1} = d^{s-1}, D_s = d_s^{\,\prime} \, ] \quad \text{is invariant in $d_s^{\,\prime}$.}
    \]
\end{enumerate}
\end{assumption}   

\noindent
Assumption \ref{assn:mse} in either form is based on the sequential exchangeability condition in \cite{robins1986g} when the covariate history consists only of lagged outcomes.\footnote{
    Other names for Assumption 2b in the literature include ``conditional exchangeability'', ``sequential ignorability'', and ``conditional parallel trends''.}

Assumption \ref{assn:mse} only restricts a stream of potential outcomes where the first $s-1$ treatment indices conform to the realized history of treatments $D^{t-1}=d^{t-1}$; it imposes conditional independence between present treatment and the entire stream of counterfactual potential outcomes in the future. 
Mean sequential exchangeability (Assumption \ref{assn:mse}b) relaxes full sequential exchangeability (Assumption \ref{assn:mse}a) in two ways. 
First, the \textit{joint} independence of treatment with the vector of present and future potential outcomes is replaced with \textit{marginal} independence of treatment with potential outcomes in each future period. 
Second, such marginal independence is further relaxed to marginal \textit{mean} independence. 
Nonetheless, each version is consistent with random assignment of treatment conditional on the {\it observed} history.

It is well-known that sequential exchangeability is sufficient for identifying a variety of treatment effects (e.g., \cite{robinshernan2009chapter}). 
Instead, one of our objectives in \Cref{sec:causal-iv-se} is to show when and how this condition also leads to causal interpretations of alternative, simple estimands in dynamic panel data models for observed outcomes.

\vspace{.25cm}

\begin{remark}
    A useful observation is that Assumption \ref{assn:mse}b, stated in terms of the \textit{levels} of potential outcomes $Y_t (d^t)$, is equivalent to an analogous assumption on the \textit{trends} in potential outcomes $Y_t (d^t) - Y_{t-1} (d^{t-1})$. This provides insights into comparisons between Assumption \ref{assn:mse}b and the parallel trends assumption used in the difference-in-differences (DiD) literature. (See \Cref{sec:evaluating}.) This equivalence is formalized in the next proposition.
\end{remark}

\vspace{.25cm}

\begin{proposition} \label{thm:mse-trend-equivalence}
    Assumption \ref{assn:mse}b is equivalent to: for any $d^{\,t} \in \{0,1\}^t$,
    \[
    \E[ Y_t (d^t) - Y_{t-1} (d^{t-1}) \mid  Y^{s-1}, D_s = d_s^{\,\prime}, D^{s-1} = d^{s-1}]
    \; \text{is invariant in $d_s^{\,\prime}$ for $t \geq s \geq 1$,}
    \]
    where $d^{s}$ is a vector consisting of the first $s$ components in $d^{\,t}$.
\end{proposition}

\vspace{.25cm}

It is straightforward to illustrate Proposition \ref{thm:mse-trend-equivalence}  when $T=2$, in which case the alternative trend formulation is: for any $(d_1,d_2)\in\{0,1\}^2$, 
\begin{align*}
    \E[ Y_1 (d_1) - Y_0 \mid  Y_0, D_1] &= \E[ Y_1 (d_1) - Y_0 \mid  Y_0]; \\
    \E[ Y_2 (d_1, d_2) - Y_1 (d_1) \mid  Y_0, D_1] &= \E[ Y_2 (d_1, d_2) - Y_1 (d_1) \mid  Y_0] ; \\
    \E[ Y_2 (d_1, d_2) - Y_1 (d_1) \mid  Y_0, Y_1, D_1 = d_1, D_2] &= 
    \E[ Y_2 (d_1, d_2) - Y_1 (d_1) \mid  Y_0, Y_1, D_1 = d_1].
\end{align*}

In one way, sequential exchangeability is strong in that it is required to hold for all potential outcomes, both treated and untreated, whereas the parallel trends assumption only restricts the trends in untreated outcomes. 
On the other hand, sequential exchangeability is required to hold conditional on observed history --- including past outcomes --- while parallel trends do not condition on past outcomes. 

\section{Causal Inference via Dynamic Panel Data Methods}
\label{sec:DPR}

Our starting point is the following equation:
\begin{equation}
\label{eq:dyn-panel}
    Y_{it} = \beta D_{it} + \gamma Y_{it-1} + \theta_t + \alpha_i  + \varepsilon_{it},
\end{equation}
where $\gamma, \beta$ are constant parameters, $\alpha_i$ are time-invariant individual fixed effects, $\theta_t$ are deterministic time trends, and $\varepsilon_{it}$ are time-varying idiosyncratic errors. 
We investigate how the standard econometric estimators for the coefficients can be related to the moments of heterogeneous treatment effects and trends in potential outcomes. 

In particular, \eqref{eq:dyn-panel} is natural for empiricists as a reduced-form equation for a series of \textit{observed} outcomes.  
For example, an intuitive interpretation is an extension of a static panel data model in which the researcher also wishes to control for lagged outcomes. 
A challenge for estimating \eqref{eq:dyn-panel} is the concern about endogenous lagged outcomes, due to the unobservable fixed effects $\alpha_i$.
The ordinary least-squares (OLS) estimators and fixed-effect estimators using either ``within'' transformation or first-differences are generally inconsistent as the number of units grows while the number of periods is fixed (e.g., \citep{baltagi2021dynamic}).

\cite{andersonhsiao82} and \cite{arellanobond91} proposed IV and GMM estimators that consistently estimate the coefficients under suitable moment and rank conditions, using lagged outcomes and covariates as instruments. 
Specifically, first-differencing \eqref{eq:dyn-panel} implies:
\begin{equation}
\label{eq:fd}
    \Delta Y_{it}=
    \beta \Delta D_{it} + \gamma \Delta Y_{it-1} + \Delta \theta_{t} +
    \Delta \varepsilon_{it}\ \quad \text{for } 2 \leq t \leq T \text{,}
\end{equation}%
where $\Delta$ denotes the first-difference operator, e.g., $\Delta Y_{it}\equiv Y_{it}-Y_{it-1}$.
\cite{andersonhsiao82} estimated the coefficients via IV regression, using $Y_{it-2}$ as instruments for $\Delta Y_{it-1}$ when $D_{it}$ is \emph{strictly} exogenous. 
\cite{arellanobond91} proposed a GMM estimator when $D_{it}$ is \emph{sequentially} exogenous (a.k.a., \emph{pre-determined}), using earlier lagged outcomes \emph{and} treatments as instruments for $\Delta Y_{it-1}$ and $\Delta D_{it}$.

It is important to emphasize that we do not interpret the observed outcome equation \eqref{eq:dyn-panel} as a causal model per se. 
Rather, our goal in \Cref{sec:causal-iv-se} below is to obtain a causal interpretation of the probability limit to which a first-differenced instrument-based estimator converges. We refer to this limit as the {\it (first-differenced) instrumental-variable (IV) dynamic panel estimand}, or simply the IV estimand when the context is clear.
We show how this estimand relates to the distribution of heterogeneous treatment effects under non-parametric assumptions on the {\it potential} outcome process.
This exercise is analogous to the interpretation of two-way fixed effect regressions under parallel trends assumptions, e.g., \cite{DS2020} and \cite{goodman2021difference}. 

\subsection{Causal Interpretation of First-Differenced IV Estimands} 
\label{sec:causal-iv-se}

We study the causal interpretation of a first-differenced IV estimand in the three-period setting $t=0,1,2$. 
Dropping individual subscripts, we define covariates and instrumental variables as the respective row vectors
\[
    X \equiv (\Delta D_2,\Delta Y_1,1),
    \qquad
    Z \equiv \E(X\mid Y_0,D_1).
\]
Throughout this section, we maintain the following regularity condition. 

\vspace{.25cm}

\begin{assumption}[IV Rank] 
\label{assn:iv-rank}
$\E(Z'X)=\E(Z'Z)$ is non-singular.
\end{assumption}

\vspace{.25cm}
\noindent
The estimand we study is the unique solution to the unconditional moment equations
\begin{equation}
\label{eq:fd-iv-moments}
    \E\!\left[
        Z'\left(
            \Delta Y_2
            -\widetilde{\beta}\Delta D_2
            -\widetilde{\gamma}\Delta Y_1
            -\widetilde{\Delta\theta_2}
        \right)
    \right]=0,
\end{equation}
or, equivalently,
\begin{equation} \label{eq:fd-iv}
    (\, \widetilde\beta,\, \widetilde\gamma, \widetilde{\Delta \theta_2} \,)^{\,\prime} = 
    \E(Z'X)^{-1}\E(Z'\Delta Y_2).
\end{equation}
The estimand in \eqref{eq:fd-iv} is the population counterpart of a first-differenced GMM estimator using instruments $\E(X\,\vert\,Y_0,D_1)$, building on the approach of \cite{arellanobond91}. 
This instrument choice provides an optimal-instrument benchmark.
Specifically, suppose observed outcomes satisfy \eqref{eq:dyn-panel} and the sequential exogeneity condition $\E(\Delta\varepsilon_2\,\vert\,D_1,Y_0) = 0$. 
Then $\E(X\,\vert\,Y_0,D_1) $ is an optimal instrument vector for forming unconditional GMM moments, provided $\Delta \varepsilon_2$ is conditionally homoskedastic \citep{chamberlain1987efficiency}.%
\footnote{ 
    We can also motivate the estimand in \eqref{eq:fd-iv} as 2SLS under an unconditional orthogonality assumption: $ \E[\,(1,Y_0,D_1)^{\,\prime}\Delta\varepsilon_2\,] = 0 $.
    Our estimand coincides with population 2SLS using instruments $(1,Y_0,D_1)$ when both $\E(\Delta D_2 \,\vert\, Y_0,D_1)$ and $\E(\Delta Y_1\,\vert\, Y_0,D_1)$ are affine in $(Y_0,D_1)$. If $(Y_0,D_1)$ has finite support, our estimand coincides with population 2SLS using a saturated set of history indicators. 
    In general, however, our estimand can differ from population 2SLS using the raw lag-level instruments.
}

One natural question to ask is the following: Are there any data generating processes of {\it potential} outcomes and treatments that would lead to a model for {\it observed} outcomes that conforms to the functional form in (\ref{eq:dyn-panel}) \emph{and} satisfy the identifying conditions, so that $\widetilde \beta = \beta$? 
The answer is positive; we present an example in \Cref{apx:ARPO}, where potential outcomes follow a first-order autoregressive process with a constant treatment effect.

For the rest of Section \ref{sec:causal-iv-se}, we address a different question, namely, whether the IV estimand $\widetilde{\beta}$ above admits a causal interpretation in general, without assuming that observed outcomes follow the linear dynamic panel model. We explore conditions (on potential outcomes and treatment assignment) under which the answer is positive. 
While doing so, we maintain that the data generating process of observed outcomes {\it does not} necessarily conform to (\ref{eq:dyn-panel}) or the identifying conditions in \mbox{\cite{arellanobond91}}.

Define history-dependent trends in potential outcomes:
\begin{equation}
\label{eq:potential-trend}
    \delta_{1}\equiv Y_{1}(0)-Y_{0} ; \quad 
    \delta _{2}(d_{1}) \equiv Y_{2}(d_{1},0)-Y_{1}(d_{1})
\end{equation}
and the treatment effects:
\begin{equation*}
    \tau_{1}\equiv Y_{1}(1)-Y_{1}(0)  
    ; \quad
    \tau _{2}(d_{1}) \equiv Y_{2}(d_{1},1)-Y_{2}(d_{1},0). 
    \label{causal_trends} 
\end{equation*} 
The trends of observed outcomes are decomposed as:%
\begin{align}
\Delta Y_{1} 
&= Y_{1}(D_{1})-Y_{0}=\delta _{1}+D_{1}\tau _{1}\text{;}  
\label{eq:Y1-decompose} \\
\Delta Y_{2} 
&= Y_{2}(D_{1},D_{2})-Y_{1}(D_{1})=\delta _{2}(D_{1})+D_{2}\tau
_{2}(D_{1})\text{.} 
\label{eq:Y2-decompose}
\end{align}

The next lemma invokes the Frisch-Waugh-Lovell (FWL) Theorem and the decomposition of $\Delta Y_2$ in \eqref{eq:Y2-decompose} to express the IV estimand $\widetilde \beta$ in \eqref{eq:fd-iv} in terms of causal and non-causal components. 
Let $W $ be the residuals in the linear projection of $Z_1 \equiv \E(\Delta D_{2}\mid Y_{0},D_{1})$ onto $Z_{-1} \equiv (\E(\Delta Y_{1}\mid Y_{0},D_{1}),1)$:
\begin{equation}
\label{eq:weights}
W \equiv w (Y_0,D_1) \equiv Z_1 - Z_{-1} \E( Z_{-1}' Z_{-1})^{-1} \E( Z_{-1}' Z_1)
\end{equation}
where the second equality clarifies that $W$ is a function of $(Y_0, D_1)$.
Under the maintained moment and rank conditions in Assumptions \ref{assn:po-treatments} and \ref{assn:iv-rank}, we have $\E(W^2) > 0$. 

\vspace{.25cm}

\begin{lemma} 
\label{thm:iv-fwl}
    Suppose Assumptions \ref{assn:po-treatments} and \ref{assn:iv-rank} hold, and that the expectations below exist and are finite. 
    The IV estimand $\widetilde \beta$ in \eqref{eq:fd-iv} is:
    \begin{equation} \label{eq:iv-fwl}
        \frac{ \E( W \Delta Y_2 ) }{ \E( W^2)} = 
        \E \left[ \frac{ W^2 }{\E(W^{2})} \tau_2 (D_1) \right]  + 
        \frac{C_\tau + C_\delta}{\E(W^2)},
    \end{equation}
    where: 
    \begin{equation}
    \label{eq:confound1}
        C_\tau \equiv \E[ W (D_2 - W) \tau_2 (D_1)],
        \qquad
        C_\delta \equiv \E[ W \delta_2 (D_1) ] 
    \end{equation}
\end{lemma}

\bigskip 
Lemma \ref{thm:iv-fwl} clarifies the basis for, and impediments to, causal interpretations of the IV estimand $\widetilde\beta$ with no other conditions than Assumptions \ref{assn:po-treatments} and \ref{assn:iv-rank}; it does not invoke sequential exchangeability in Assumption \ref{assn:mse}.
The first term on the right in (\ref{eq:iv-fwl}) is a convex-weighted aggregate of individually heterogeneous treatment effects $ \tau_2(\cdot) $ in the second period, which alone would provide a clear causal interpretation of $\widetilde \beta$. 
The second term consists of two sources potentially biasing such an interpretation. 

The first source of bias, involving $C_\tau$, is a correction to the convex weights based on the relationship between the FWL weights $W$ and treatment $D_2$, where the overall linear combination of $\tau_2 (D_1)$ in \eqref{eq:iv-fwl} can involve negative weights or weights that fail to sum to one.
The second source of bias, involving $C_\delta$, is a linear combination of history-dependent trends $\delta_2 (\cdot)$.
The first source $C_\tau$ parallels a point by \cite{DS2020} and \cite{goodman2021difference} in the causal interpretation of two-way fixed-effect (TWFE) regressions in static panel data models; yet, both sources of bias can persist even if treatment effects $\tau_2 (\cdot)$ are constant, unlike the standard static TWFE decomposition under parallel trends.
The trend term $C_\delta$ is a source of contamination bias, as in \cite{goldsmith2024contamination}.

To endow this decomposition of the IV estimand $\widetilde\beta$ in \eqref{eq:iv-fwl} with a causal interpretation, we impose 
further restrictions on heterogeneous treatment effects in addition to those on selection in Assumption \ref{assn:mse}. 
In particular, these assumptions suffice to eliminate the bias terms, i.e., $C_\tau = C_\delta = 0$ in \eqref{eq:iv-fwl}. 
In a sense we make precise in our results, the conditions are also necessary to guarantee a causal interpretation. 

\vspace{.25cm}

\begin{assumption}[Limited Heterogeneity of Trends and Treatment Effects]
\label{assn:limited-heterogeneity}
For all $d_1 \in \{ 0,1 \}$, the following equalities hold almost surely:
\begin{align}
    \E [\delta _{1}\mid Y_{0}]
    &=
    \E [\delta _{1}],
    \label{eq:invar_trend_1} \tag{a} \\
    \E [\tau _{1}\mid Y_{0}]
    &= 
    \E [\tau_{1}], 
    \label{eq:invar_te_1} \tag{b} \\
    \E[\delta _{2}(d_{1}) \mid Y_{0}]
    &=
    \E[\delta _{2}(d_{1})],
    \label{eq:invar_trend_2} \tag{c} \\
    \E[\tau _{2}(d_{1})\mid Y_{0}, Y_{1} (d_1)]
    &=
    \E[\tau _{2}(d_{1})].  
    \label{eq:invar_te_2} \tag{d}
    \end{align}%
\end{assumption}

\noindent 
Assumption \ref{assn:limited-heterogeneity} requires the first moments of treatment effects $\tau_t (\cdot)$ and untreated trends $\delta_t (\cdot)$ to be invariant in the initial condition or earlier potential outcomes; this holds, for example, when these treatment effects and untreated trends are homogeneous in the population.
Jointly with Assumption \ref{assn:iv-rank} and Assumption \ref{assn:mse}b, this also imposes a rank condition on the treatment effect $\E( \tau_1) \neq 0$ in period 1. 

In stating the following two results, let $\mathcal{F}$ denote the class of joint distributions of the full potential outcome and assignment model $(\,D_1,D_2, Y_1(\cdot),Y_2(\cdot,\cdot)\,)$ that satisfy Assumptions \ref{assn:po-treatments}, \ref{assn:mse}b, and \ref{assn:iv-rank}, and overlap: $0 < P (D_1 = 1 | Y_0) < 1$ almost surely.
Let $\mathcal{F}^* \subset \mathcal{F}$ denote the subclass that also satisfies Assumption \ref{assn:mse}a and Assumption \ref{assn:limited-heterogeneity} (a)--(c). 
Let $V \equiv (D_1, Y_0, Y_1 (\cdot), Y_2 (\cdot, 0))$ denote a random vector excluding period-2 treatment and treated outcomes $(D_2, Y_2 (\cdot, 1))$, and let $V_+ \equiv (V, Y_2 (\cdot, 1))$ denote a random vector excluding only period-2 treatment. 
Finally, let $\mathcal{F}_V$ denote the class of marginal distributions $F_V$ of $V$ induced by $\mathcal{F}$, and analogously for $\mathcal{F}_{V_+}^*$.

\vspace{.25cm}

\begin{proposition}[Eliminating Trend Contamination]
\label{thm:lh1}
Fix a marginal distribution $G \in \mathcal{F}_V$.
The following are equivalent: 
\begin{enumerate}
    \item $G$ satisfies limited heterogeneity conditions (a)-(c) in Assumption \ref{assn:limited-heterogeneity}.
    \item  Every $F \in \mathcal{F}$ with $F_V = G$ satisfies: 
\begin{equation}
\label{eq:beta-no-contam}
    \tilde{\beta} = \E \left[ 
    \frac{W D_2}{E(W^2)}\tau_2 (D_1) 
    \right], 
    \quad
    \E \left[ 
    \frac{W D_2}{E(W^2)}
    \right] = 1
\end{equation}
    \item Every $F \in \mathcal{F}$ with $F_V = G$ and with constant and strictly positive period-2 treatment effects 
    satisfies $\tilde{\beta} \geq 0$. 
\end{enumerate}
\end{proposition}

\vspace{.25cm}

Proposition \ref{thm:lh1} characterizes weak causality, in the sense of \cite{blandhol2026tsls}, for homogeneous period-2 effects uniformly over joint distributions in $\mathcal{F}$ sharing the fixed marginal $G$.
If conditions (a)-(c) of Assumption \ref{assn:limited-heterogeneity} fail under $G$, some such joint distribution produces a negative coefficient despite strictly positive, constant period-2 effects, ruling out any convex treatment effect representation for that distribution. 
Under conditions (a)-(c), untreated trend contamination vanishes, $C_\delta = 0$, and the projection residual simplifies to:
\begin{equation}
\label{eq:weights-convex}
    W = \E( D_2 \mid Y_0, D_1) - \E( D_2 \mid D_1).
\end{equation}
The resulting treatment effect weights sum to one, so constant effects are recovered exactly. 
These conclusions require no restriction on period-2 effect heterogeneity. 
With such heterogeneity, these weights can be negative, and $C_\tau$ need not vanish. 
Under full sequential exchangeability, the next proposition characterizes the additional restriction on period-2 effect heterogeneity needed to eliminate $C_\tau$ uniformly over period-2 assignment rules, holding the distribution of $V_+$ fixed and yielding a squared-weight convex representation.

\vspace{.25cm}

\begin{proposition}[Convex Representation]
\label{thm:convex-te}
Fix a marginal distribution $G \in \mathcal{F}_{V_+}^*$. The following are equivalent: 
\begin{enumerate}
    \item $G$ satisfies limited heterogeneity condition (d) in Assumption \ref{assn:limited-heterogeneity}. 
    \item Every $F \in \mathcal{F}^*$ with $F_{V_+} = G$ satisfies:
    \begin{equation}
    \label{eq:te-convex}
        \tilde{\beta} = \E \left[ 
            \frac{W^2}{\E(W^2)} \tau_2 (D_1)
        \right]
    \end{equation}
    \item There exists a finite constant $c_G$ such that $| \widetilde{\beta} (F) | \leq c_G$ for all $F \in \mathcal{F}^*$ with $F_{V_+} = G$.
\end{enumerate}
\end{proposition}

\vspace{.25cm}

Proposition \ref{thm:convex-te} --- specifically, the equivalence between statements 1 and 2 --- characterizes the squared-weight convex representation uniformly over period-2 assignment rules. 
When Assumption \ref{assn:limited-heterogeneity} (d) fails, 
a particular assignment rule could still yield this or another convex representation. 
However, as used in proving the contrapositive of $(3 \implies 1)$, even sequentially exchangeable assignment of $D_2$ can exploit any predictable effect heterogeneity and amplify the distortion through a weak first stage, producing arbitrarily large or small IV estimands. 
If effects are bounded on either side, some compatible assignment rule therefore yields a coefficient outside their range, ruling out \textit{any} convex treatment-effect representation.%
\footnote{
If treatment effects are unbounded in both directions, unrestricted weights can always be chosen to match a finite coefficient. However, such ex post existence does not provide a prespecified weighting interpretation.
}

We conclude with two comments on the scope of the preceding analysis. 
First, these characterizations concern the specified regression, which includes current treatment and the lagged outcome. 
Although our potential outcome framework allows effects of past treatments, the simplest regression does not include treatment lags separately. 
Adding such lags changes the estimand, and it is not clear whether doing so could eliminate particular sources of distortion due to contamination biases. 
Characterizing these alternative specifications is, however, beyond the scope of our analysis. 

Second, with sequential exchangeability invoking ``conditional unconfoundedness'' given past treatments and observed outcomes, one may use methods such as the inverse-propensity-weighted estimator to infer the average treatment effect $\E[\tau_2(D_1)]$. 
It is also  well-known that counterfactual means $E[ Y_2 (d_1, d_2)]$ and average treatment effects $\E[\tau_2(d_1)]$ are identified for $d_1, d_2\in\{0,1\}$ under sequential exchangeability in Assumption \eqref{assn:mse} alongside relevant overlap conditions; see, e.g., the $g$-formula in \citep{robinshernan2009chapter}.
Instead, our goal in \Cref{sec:causal-iv-se} has been to better understand what is estimated by first-differenced IV regression when the latent potential outcomes are nonparametric with heterogeneous, intertemporal treatment effects, even under this idealized restriction on selection. 
This parallels analyses of TWFE and TSLS coefficients under assumptions that already identify causal effects  \citep{DS2020,goodman2021difference,blandhol2026tsls}.

\subsection{Causal Inference without Sequential Exchangeability}  \label{sec:causal-no-se}

Sequential exchangeability ($\text{Assumption \ref{assn:mse}}$) fails in observational settings when unobserved heterogeneity confounds treatment selection, even after controlling for past treatment and outcomes. We now study how the dynamic panel data models can identify average treatment effects on the treated without this assumption.

As a starting point, we consider a structural causal model with homogeneous treatment effects in \Cref{apx:ARPO}, where \emph{(i)} potential outcomes $Y_{it}(d^t)$ for all $d^t$ follow a first-order autoregressive process with a constant auto-regressive coefficient, and \emph{(ii)} shocks in potential outcomes are serially uncorrelated conditional on past treatment. 
We show in that case the IV estimand in \eqref{eq:fd-iv} identifies the treatment effect and auto-regressive coefficient, even when sequential exchangeability fails. 
This works because the structure in (i) allows us to represent potential outcomes via observed ones, with the treatment effect captured by a single coefficient, and the serial uncorrelation in (ii) means the shocks in observed outcomes satisfy sequential exogeneity as required in \cite{arellanobond91}.

This section concentrates on a general model with \emph{heterogeneous} treatment effects that are time-varying and history-dependent:
\begin{equation*} %\label{eq:arupo-hte} 
    \tau_{it}(d^{t-1}) \equiv Y_{it}(d^{t-1},1) - Y_{it}(d^{t-1},0).
\end{equation*}
Unlike the benchmark model with constant effects in \Cref{apx:ARPO}, here we introduce  autoregressive structure only on \emph{untreated} potential outcomes.
For example, a firm's untreated productivity may follow a persistent baseline process, while exporting generates gains that vary across firms and depend on previous export experience. Restricting the untreated process need not impose the same homogeneous autoregressive process on those gains. 
We separate this baseline restriction from the restriction on gains needed for identification. 

\vspace{.5cm}

\begin{assumption} \label{assn:arupo}
    For any $d^{t-1}\in\{0,1\}^{t-1}$, the untreated potential outcome at $t$ follows a first-order autoregressive model:
    \begin{equation}\label{eq:ARUPO}
        Y_{it}(d^{t-1},0) = \gamma Y_{it-1}(d^{t-1}) + \alpha_i + \theta_t + \varepsilon_{it}(d^{t-1}) \text{ for } t\geq 1, 
    \end{equation}
    with  initial condition $Y_{i0} (d^0) \equiv Y_{i0}$, where $\alpha_i$ are individual fixed effects, $\theta_t$ time effects, and $\varepsilon_{it}(d^{t-1})$ are time-varying shocks in untreated potential outcomes s.t.
    \begin{equation} \label{assn:ARUPO_1}        
        \E[\,\varepsilon_{it}(d^{t-1}) \mid D_i^t=d^t,Y_{i0},\tau_i^{t-1}(d^{t-2}),\varepsilon_i^{t-1}(\cdot),\alpha_i\,] = 0 \; \text{ for } t\geq 1 \text{, } \;  d_t\in\{0,1\}.
    \end{equation}
\end{assumption}

\vspace{.25cm}

In \eqref{assn:ARUPO_1}, $\varepsilon_i^{s}(d^{s-1})\equiv(\,\varepsilon_{i\tilde s}(d^{\tilde s-1}):1\le \tilde s\le s\,)$ collects the shocks along the path $d^{s-1}$, and $\varepsilon_i^{s}(\cdot)$ collects $\varepsilon_i^{s}(d^{s-1})$ over all $d^{s-1}\in\{0,1\}^{s-1}$, as in $Y_i^t(\cdot)$, with $\varepsilon_i^{0}(\cdot)$ vacuous; $\tau_i^{s}(d^{s-1})$ and $\tau_i^{s}(\cdot)$ are defined likewise.
\Cref{assn:arupo} does not rule out dependence of $D_{it}$ on $Y_{it-1}$ or on $(D_i^{t-1},Y_i^{t-2})$. It posits the potential shocks $\varepsilon_{it}(\cdot)$ are mean independent of treatment history up to $t$ and latent variables contributing to past observed outcomes up to $t-1$. 

Under this assumption, \emph{observed} outcomes can be written as
\begin{equation*}
    Y_{it}\equiv Y_{it}(D_i^t)= \tau_{it} D_{it} + \gamma Y_{it-1} + \alpha_i + \theta_t + \varepsilon_{it} \text{ for } t\geq 1,
\end{equation*}
where $\tau_{it},\varepsilon_{it}$ are shorthand for $\tau_{it}(D_i^{t-1}),\varepsilon_{it}(D_i^{t-1})$ respectively. First-differencing gives:
\begin{equation} \label{eq:fd-obs}
    \Delta Y_{it} = \tau_{it} D_{it} - \tau_{it-1} D_{it-1} + \gamma \Delta Y_{it-1} + \Delta \theta_t + \Delta \varepsilon_{it} \text{ for } t\geq 2.
\end{equation}
By \eqref{assn:ARUPO_1} in \Cref{assn:arupo}, \(\E(\varepsilon_{it} \vert D_i^t,Y_i^{t-1}) = 0\) for all $t\geq 1$, because the observed history $Y_i^{t-1}$ is a function of \( (D_i^{t-1},Y_{i0},\tau_i^{t-1}(D_i^{t-2}),\varepsilon_i^{t-1}(D_i^{t-2}),\alpha_i) \).
This implies the history of observed treatment and outcomes is sequentially exogenous:
\begin{equation} \label{eq:seq-exo}
    \E(\Delta\varepsilon_{it} \vert H_{it-1}) = 0 \text{ for } t\geq 2, \; \text{ where } \; H_{it-1} \equiv ( D_i^{t-1},Y_i^{t-2} ).
\end{equation}
This is reminiscent of the moment condition implied by sequentially exogenous regressors in \cite{arellanobond91}. Yet the model above is more general: the first-difference in (\ref{eq:fd-obs}) contains random coefficients due to heterogeneous treatment effects $\tau_{it} D_{it} - \tau_{it-1} D_{it-1}$.

A causal parameter of interest is the average treatment effect of the {\it first-time treated}: 
\begin{equation} \label{eq:atft}
    \mathrm{ATFT}_t(y_0)\equiv \E[\tau_{it} \vert D_{it}=1, \mathcal{S}_{t-1} (y_0)] \text{ for } t\geq 1,
\end{equation}
where \( \mathcal{S}_{t-1} (y_0) \equiv \{D_i^{t-1}=0^{t-1}, Y_{i0}=y_0 \} \); on this event, $\tau_{it}=\tau_{it}(0^{t-1})$, which we denote by $\taus$.
(Appendix \ref{apx:gen-arupo} extends the logic to identify the average treatment effect under a general history of treatments $d^{t-1}\neq 0^{t-1}$.)

Our next result shows $\mathrm{ATFT}_t$ can be recovered by a first-differenced IV estimand conditional on the event $\mathcal S_{t-1}(y_0)$ (that the unit is endowed with an initial condition $y_0$ and has not been treated through $t-1$).
This result requires the following condition on how treatments relate to latent components of the model in \Cref{assn:arupo}, including the heterogeneous treatment effects.

\vspace{.25cm} 

\begin{assumption} \label{assn:MI-hte} 
For any initial condition $y_0$  and $t\geq 1$,
    \[
        \E[\,\taus\mid D_{it}=1, \mathcal S_{t-1} (y_0), \varepsilon_i^{t-1}(\cdot), \alpha_i\,]
        =   \E[\,\taus\mid D_{it}=1, \mathcal S_{t-1} (y_0)\,].
    \]
\end{assumption}

This assumption posits that, among units treated for the first time at $t$ with the same initial condition, the fixed effects $\alpha_i$ and past shocks $\varepsilon_i^{t-1}(\cdot)$ carry no information about the conditional mean of the treatment effect $\taus$. It is silent about units untreated at $t$. This is all that \Cref{thm:IV-arupo} below requires, because $\taus$ enters the first difference \eqref{eq:fd-obs} only through the product $\taus D_{it}$. 
It does not rule out dependence of $D_{it}$ on the observed history, the fixed effects, or past shocks, nor does it rule out dependence on factors that are informative about $\taus$ itself; in particular, it does not require sequential exchangeability, as the following example illustrates.

\vspace{.25cm}

\begin{example}[\Cref{assn:MI-hte} under endogenous treatment selection] \label{ex:two}
We provide an example of treatment rules under which \Cref{assn:MI-hte} holds while sequential exchangeability can fail. 
The treatment rule selects on the \emph{levels} of past observed outcomes and private information about the latent components in untreated potential outcomes, as well as on factors informative about the treatment effects, or ``gains''. 
For $s\le t$, suppose
\[
    D_{is}=A_{is}B_{is},\qquad
    A_{is}=\mathbf{1}\{\pi_s(Y_i^{s-1},v_{is})>0\},\qquad
    B_{is}=\mathbf{1}\{\omega_{is}>0\},
\]
where $\pi_s(\cdot)$ is unrestricted, $v_{is}$ is a private signal that may depend on $\alpha_i$ and $\varepsilon_i^{s-1}(\cdot)$, 
and $\omega_{is}$ is a scalar private signal that may be correlated with the gain $\tau_{is}$. We maintain \Cref{assn:arupo}, which holds here if the current shock $\varepsilon_{it}(\cdot)$ has conditional mean zero given $Y_{i0}, \alpha_i$, the past shocks and treatment effects, and the signals $(v_i^t, \omega_i^t)$. Suppose also
\begin{equation}\label{eq:ex-offer-takeup}
    (\taus,\omega_{it})\ \perp\ \big(\alpha_i,\varepsilon_i^{t-1}(\cdot),v_i^t,\omega_i^{t-1}\big)\ \big|\ Y_{i0},
\end{equation}
where $v_i^t\equiv(v_{i1},\ldots,v_{it})$ and $\omega_i^{t-1}\equiv(\omega_{i1},\ldots,\omega_{it-1})$. 

For instance, a physician \emph{offers} treatment ($A_{it}=1$) based on the patient's outcome history $Y_i^{t-1}$ and a signal $v_{it}$ about the patient's baseline health $\alpha_i$; 
the patient \emph{accepts} ($B_{it}=1$) based on her own signal $\omega_{it}$, such as her perceived net benefit of treatment, which combines a diagnostic test result, newly drawn each period, that predicts her response with her cost of treatment. 
Nobody observes the gain itself; selection on gains arises because the test result is correlated with it.
Condition \eqref{eq:ex-offer-takeup} says that, given $Y_{i0}$, the current gain and the patient's signal are unrelated to the baseline health, the past shocks, the physician's signals through $t$ and the patient's past signals.
The special case $\omega_{is}>0$ almost surely, in which every offer is accepted, is selection on levels alone.

\Cref{assn:MI-hte} holds in this example. By \eqref{eq:ARUPO}, the untreated path $Y_i^{t-1}(0^{t-1})$ is a function of $(Y_{i0},\alpha_i,\varepsilon_i^{t-1}(\cdot))$, and it coincides with the observed history on $\mathcal S_{t-1}(y_0)$. 
Let $\Lambda_{it}\equiv(\alpha_i,\varepsilon_i^{t-1}(\cdot),v_i^{t},\omega_i^{t-1})$. 
For each $s\le t$ and on the event $\{D_i^{s-1}=0^{s-1}\}$, the offer $A_{is}$ is a function of $Y_i^{s-1}(0^{s-1})$ and $v_{is}$, while the take-up $B_{is}$ is a function of $\omega_{is}$. 
Proceeding recursively from $s=1$, both $\mathbf{1}\{\mathcal S_{t-1}(y_0)\}$ and $A_{it}\mathbf{1}\{\mathcal S_{t-1}(y_0)\}$ are therefore functions of $(Y_{i0},\Lambda_{it})$. Consequently, $D_{it}\mathbf{1}\{\mathcal S_{t-1}(y_0)\}=\mathbf{1}\{\omega_{it}>0\}\,A_{it}\mathbf{1}\{\mathcal S_{t-1}(y_0)\}$ depends on the current pair $(\taus,\omega_{it})$ only through $\mathbf{1}\{\omega_{it}>0\}$.
Because $(\taus,\omega_{it})\perp \Lambda_{it}\mid Y_{i0}$ by \eqref{eq:ex-offer-takeup}, both sides of the equality in \Cref{assn:MI-hte} equal $\E(\taus\mid \omega_{it}>0,\,Y_{i0}=y_0)$, 
which is $\mathrm{ATFT}_t(y_0)$. 

Sequential exchangeability can fail in this example because of selection on levels or gains. First, the private signal $v_{it}$ may be informative about $\alpha_i$ beyond the observed history, so that $\E[\,Y_{it}(0^{t})\mid Y_i^{t-1},D_i^{t-1}=0^{t-1},D_{it}\,]$ can vary with $D_{it}$.\footnote{
By \eqref{eq:ARUPO} and \eqref{assn:ARUPO_1}, this conditional mean equals $\gamma Y_{it-1}+\theta_t+\E(\,\alpha_i\mid Y_i^{t-1},D_i^{t-1}=0^{t-1},D_{it}\,)$. Given $Y_i^{t-1}$ and $D_i^{t-1}=0^{t-1}$, $A_{it}$ varies only through $v_{it}$, so the last term can vary with $D_{it}$ when $v_{it}$ is informative about $\alpha_i$ beyond that history.}
Second, treatment can correlate with the gain: by \eqref{eq:ex-offer-takeup}, $\E(\taus\mid Y_i^{t-1},D_i^{t-1}=0^{t-1},D_{it}=1)=\E(\taus\mid \omega_{it}>0,Y_{i0})$.
Hence $\E[\,Y_{it}(0^{t-1},1)\mid Y_i^{t-1},D_i^{t-1}=0^{t-1},D_{it}\,]$ can vary with $D_{it}$ when the signal $\omega_{it}$, and hence acceptance, is informative about the mean gain beyond observed history.
 \hfill $\square$
\end{example}

\vspace{.25cm}

Let $\varrho_{t-1}(\cdot)$ be a user-chosen function of history $H_{it-1}$ and write $H^*_{it-1} \equiv \varrho_{t-1}(H_{it-1})$. 
Next, for $Y_{i0} = y_0$, consider an IV regression  of $\Delta Y_{it}$ on $X_{it} \equiv (\Delta D_{it}, \Delta Y_{it-1}, 1) $ \emph{conditional on} the event $\mathcal S_{t-1} (y_0)$, using the following instruments:
\[ Z_{it} (y_0) \equiv \E[\,X_{it}\mid H^*_{it-1}, \mathcal S_{t-1}(y_0)\,],\] 
which satisfy $\E[\Delta\varepsilon_{it} \mid Z_{it}, \mathcal S_{t-1} (y_0)]=0$ because of (\ref{eq:seq-exo}), as $Z_{it}(y_0)$ and the event $\mathcal S_{t-1}(y_0)$ are determined by $H_{it-1}$.

To simplify notation, we suppress the realized value of the initial condition $y_0$ in $\mathcal S_{t-1}(y_0)$, the instrument $Z_{it}(y_0)$, and the target parameter $\mathrm{ATFT}_t(y_0)$. Denote the estimand of this (conditional) first-differenced IV regression by:
\begin{equation}  \label{eq:IV-estimand}
    (\, \widetilde \beta_t, \, \widetilde \gamma_t, \, \widetilde{\Delta \theta_t} )'
    \equiv 
    \E(\,Z_{it}' X_{it} \mid \mathcal S_{t-1}\,)^{-1}\E(\,Z_{it}' \Delta Y_{it} \mid \mathcal S_{t-1}\,). 
\end{equation}
Our next result shows under Assumptions \ref{assn:arupo} and \ref{assn:MI-hte}, the estimand of (\ref{eq:IV-estimand}) recovers $\mathrm {ATFT}_t$, as well as the state dependence $\gamma$ and the time trends $\Delta \theta_t$.

\vspace{.25cm}

\begin{proposition} \label{thm:IV-arupo}
Suppose Assumptions \ref{assn:arupo} and \ref{assn:MI-hte} hold, and $\E(\,Z_{it}'Z_{it}\,\vert\, \mathcal S_{t-1}\,)$ has full rank for $t\geq 3$. Then: 
\begin{align*}
    \widetilde \beta_t = \mathrm{ATFT}_t, \quad % \\
    \widetilde \gamma_t = \gamma, \quad % \\
    \widetilde{\Delta \theta_t} = \Delta \theta_t \; \text{ for } t \ge 3.
\end{align*}
\end{proposition}

Both $\gamma$ and $\Delta \theta_t$ are \textit{over}-identified, because these are constant parameters while the first-differenced IV estimands in \eqref{eq:IV-estimand} condition on $y_0$.
The intuition for Proposition \ref{thm:IV-arupo} is as follows.
Under \Cref{assn:MI-hte}, the heterogeneous treatment effects $\taus$ are mean independent of $Y_i^{t-2}$ conditional on $D_{it}=1$ and $\mathcal S_{t-1}$. 
This implies the conditional mean of $\Delta Y_{it}$ given $H^*_{it-1}$ and $\mathcal S_{t-1}$ is linear in $Z_{it}$, with slope coefficients $\mathrm{ATFT}_t$, $\gamma$, and an intercept $\Delta\theta_t$.
Thus, these parameters are recovered by regressing $\Delta Y_{it}$ on $Z_{it}$ conditional on $\mathcal S_{t-1}$.
By construction, such an estimand is (\ref{eq:IV-estimand}), because $ \E(\,Z_{it}' X_{it}\,\vert\, \mathcal S_{t-1}\,) = \E(\,Z_{it}' Z_{it} \,\vert\, \mathcal S_{t-1}\,)$ by the law of iterated expectations.

\vspace{.25cm}

\begin{remark} \label{rem:rank-ex2}
    In \Cref{ex:two}, the rank condition in \Cref{thm:IV-arupo} holds under the following conditions, despite the endogenous selection into observed treatment. 
    On $\mathcal S_{t-1}(y_0)$, the condition \eqref{eq:ex-offer-takeup} implies $\E(\Delta D_{it}\mid H^*_{it-1},\mathcal S_{t-1})=p_t(y_0)\,\E(A_{it}\mid H^*_{it-1},\mathcal S_{t-1})$, where $p_t(y_0)\equiv\Pr\{\omega_{it}>0\mid Y_{i0}=y_0\}$ is the acceptance probability. 
    The rank condition therefore holds if $p_t(y_0)>0$, the predicted trend $\E(\Delta Y_{it-1}\mid H^*_{it-1},\mathcal S_{t-1})$ varies with $H^*_{it-1}$, and the offer probability $\E(A_{it}\mid H^*_{it-1},\mathcal S_{t-1})$ is not an affine function of it on $\mathcal S_{t-1}$. 
\end{remark}

\vspace{.25cm}

\begin{remark} %[Estimation and inference] 
\label{rem:est-inf}
When $Y_{i0}$ is continuous, the event $\mathcal S_{t-1}$ at any realized initial condition $y_0$ has probability zero. 
As explained above, \eqref{eq:IV-estimand} is also a regression of $\Delta Y_{it}$ on $Z_{it}$ conditional on $\mathcal S_{t-1}$, with coefficients varying with $y_0$. 
Therefore, it can be estimated by a two-step sieve regression, using the subsample of units that are untreated through $t-1$.
First, estimate $Z_{it}$ by regressing $X_{it}$ on a flexible basis (e.g., splines) in $(H^*_{it-1},Y_{i0})$.
Second, regress $\Delta Y_{it}$ by OLS on the interaction of the fitted $\widehat Z_{it}$ and a vector of sieve basis functions $p^K(Y_{i0})$, where $K$ grows with the sample size. 
The coefficient function on the first component of $\widehat Z_{it}$ estimates $\mathrm{ATFT}_t(\cdot)$.
Because $\gamma$ and $\Delta\theta_t$ do not vary with $y_0$, we only need to interact the sieve estimates of $\E(\Delta D_{it}\,\vert\,H^*_{it-1},\mathcal S_{t-1}) $ with $p^K(Y_{i0})$ in the second-step regression.%
\footnote{
    Under additional sampling, smoothness, sieve-rate, and relevance conditions, inference on $\gamma, \Delta \theta_t$ and suitable weighted averages of $\mathrm{ATFT}_t$ can use joint two-step variance formulas that account for first-stage estimation \citep{ackerberg2012practical}.
    Pointwise inference for $\mathrm{ATFT}_t$ additionally requires control of sieve approximation bias and the corresponding nonparametric variance \citep{chenliao2015sieve}. A formal asymptotic analysis for this estimator is beyond our identification results. }
\end{remark}

\vspace{.25cm}

In \Cref{apx:gen-arupo} we extend the logic behind Assumption \ref{assn:MI-hte} and Proposition \ref{thm:IV-arupo} to identify the average treatment effects for groups other than the first-time treated: 
\[\mathrm{ATT}_t(d^{t-1},y_0) \equiv \E[\,\tau_{it}(d^{t-1}) \,\vert\, D_i^t = (d^{t-1},1),Y_{i0}=y_0\,] \;\text{ for } \; d^{t-1}\neq 0^{t-1}.\]
The extension makes two changes. 
First, the IV regression in \eqref{eq:IV-estimand} is run conditional on a general history $\{D_i^{t-1}=d^{t-1},Y_{i0}=y_0\}$, with instruments defined conditional on that history. 
Second, \Cref{assn:MI-hte} is strengthened (as \Cref{assn:MI-hte2}) by also requiring current effects to be mean independent of past effects within the specified treated-history group. This is needed because, once a unit has been treated, its observed history, and hence the instruments, depend on its past effects. 
Under these conditions, we show (in \Cref{thm:arupo-g-iv}) that the slope estimands from IV regression recover $\mathrm{ATT}_t(d^{t-1},y_0)$ and the state dependence parameter $\gamma$. 
The estimation approach in \Cref{rem:est-inf} applies conditional on each history $d^{t-1}$.

\section{Sequential Exchangeability versus Parallel Trends} \label{sec:evaluating} 

Our results in \Cref{sec:DPR} complement the difference-in-differences method when applied to settings with non-absorbing treatment and intertemporal outcomes; see, e.g., \cite{chaisemartinhaultfoeuille2024dyn}.%
\footnote{
They identify event-study effects comparing a group's realized treatment path after its first change with the counterfactual path if its treatment had remained fixed at its initial value, where treatment need not be binary. 
Their normalized effects aggregate effects of current and lagged treatment doses.
Our contrast $Y_t (d^{t-1}, 1) - Y_t (d^{t-1},0)$ instead changes only current treatment, holding prior treatment fixed. 
These targets can coincide at the time of an initial switch but generally differ at later horizons. 
}
In a model with $T=2$ and heterogeneous intertemporal treatment effects, the parallel trends assumption, conditional on the initial $Y_0$, posits:
\begin{equation}\label{assmp:PT}
    \E[Y_2(0,0)-Y_1(0)\mid D_1,D_2,Y_0] \text{ is invariant in } (D_1,D_2).
\end{equation} 

In this section, we study how sequential exchangeability and parallel trends restrict the relation between treatment realization and potential outcomes. 
Our objective is to show how these conditions are justified in \emph{experimental} or \emph{observational} data. 

\subsection{Sequentially Randomized Treatments} \label{sec:seq-RT}

Sequential exchangeability is the assignment condition delivered by randomization conditional on observed history. 
Such randomization underlies experimental designs for adaptive treatment strategies, particularly in medicine, in which later interventions respond to earlier outcomes (e.g., \cite{lavoridawson2000, murphy2005}).
These designs provide one natural benchmark for our comparison: even when assignment is randomized within history strata, subsequent treatment groups need not have equal untreated trends.

Suppose the observed treatments are due to sequential randomization in each period after controlling for the history of realized treatments and observed outcomes up to that period. With $T=2$, this means:
\begin{eqnarray}
D_{1} \perp (Y_{1}(\cdot ),Y_{2}(\cdot ,\cdot )) & \mid & Y_0, \label{assum:seq-rand1} \\ 
D_{2} \perp Y_{2}(D_1,\cdot ) & \mid & Y_1,Y_0,D_1 \text{.}  \label{assum:seq-rand2}
\end{eqnarray}
The following proposition compares the implications of this assignment mechanism for sequential exchangeability and parallel trends.
For simplicity, let the initial condition $Y_{0}$ be degenerate and suppressed in conditioning events.
\bigskip

\begin{proposition}\label{pn:SRvsPT}
    Sequential randomization of treatments in (\ref{assum:seq-rand1}) and (\ref{assum:seq-rand2}) implies the (mean) sequential exchangeability in Assumption \ref{assn:mse}b, but not the parallel trends in (\ref{assmp:PT}).
\end{proposition}
\bigskip

Mean sequential exchangeability in Assumption \ref{assn:mse}b holds under (\ref{assum:seq-rand1}) and (\ref{assum:seq-rand2}) because:
\begin{eqnarray*}
    \E[Y_{2}(d_{1},d_{2})\mid D_{2},Y_{1},D_{1}=d_{1}]&=&\E[Y_{2}(d_{1},d_{2})\mid Y_{1},D_{1}=d_{1}]\text{,} \\
    \E[Y_{2}(d_{1},d_{2})\mid D_{1}]&=&\E[Y_{2}(d_{1},d_{2})]\text{,} \\
    \E[Y_{1}(d_{1})\mid D_{1}]&=&\E[Y_{1}(d_{1})]\text{,}
\end{eqnarray*}
for all $d_{1},d_{2}\in \{0,1\}$.

Parallel trends in (\ref{assmp:PT}) are not implied under (\ref{assum:seq-rand1}) and (\ref{assum:seq-rand2}). 
To see this, note:
\begin{eqnarray*}
& & \E[Y_{2}(0,0)-Y_{1}(0)\mid D_{2}=d_{2},D_{1}=0] \\
&=& \int \E[Y_{2}(0,0)-Y_{1}(0)\mid D_{2}=d_{2},Y_{1}(0)=y_{1},D_{1}=0]dF_{Y_{1}(0)}(y_{1}\mid D_{2}=d_{2},D_{1}=0)
\\
&=&\int \left[ \mu _{2}(y_{1})-y_{1}\right]
dF_{Y_{1}(0)}(y_1\mid D_{2}=d_{2},D_{1}=0)\text{,}
\end{eqnarray*}%
where $\mu _{2}(d_2,s)\equiv \E[Y_{2}(0,0)\mid  D_2 = d_2, Y_{1}(0)=s,D_{1}=0]$ does not vary with $d_2$ under (\ref{assum:seq-rand2}), and hence the first argument is suppressed. 
By construction, 
\begin{eqnarray*}
F_{Y_{1}(0)}(y\,|\,D_{2}=1,D_{1}=0) = & \frac{\int_{-\infty }^{y}p(s)dG(s)}{\int_{-\infty }^{+\infty }p(s)dG(s)}, \\
F_{Y_{1}(0)}(y\,|\,D_{2}=0,D_{1}=0) = & \frac{\int_{-\infty }^{y}[1-p(s)]dG(s)}{\int_{-\infty }^{+\infty }[1-p(s)]dG(s)}\text{,}
\end{eqnarray*}
where $p(s)\equiv \Pr \{D_{2}=1\,|\,Y_{1}(0)=s,D_{1}=0\}$, and $G(s)\equiv \Pr\{Y_{1}(0)\leq s \,|\, D_{1}=0\}=\Pr \{Y_{1}(0)\leq s\}$ under (\ref{assum:seq-rand1}). 
Thus, if $p(s)$ is non-constant in $s$, $F_{Y_{1}(0)}(y \,|\, D_{2}=d_2,D_{1}=0)$ generally varies with $d_2$, which can generate unequal mean untreated trends and violate parallel trends (\ref{assmp:PT}).

\subsection{Forward-Looking Treatment Choices with Learning} 
\label{sec:evaluating-learning}

We now turn to observational data, and show how sequential exchangeability and parallel trends are justified when rational, forward-looking units choose treatments optimally with updated beliefs/learning from past outcomes. 

Let $T=2$, and suppose the potential outcomes are:
\begin{eqnarray} \label{eq:po-choice}
    Y_{1}(d_{1}) &=&\lambda _{1}(d_{1},Y_{0},\alpha(d_1),\varepsilon_{1})\text{,}   \\
    Y_{2}(d_{1},d_{2}) &=&\lambda _{2}(d_{1},d_{2},Y_{0},Y_{1}(d_{1}),\alpha(d_2),\varepsilon_2)\text{,}  \notag
\end{eqnarray}%
with $\lambda_t(\cdot)$ being deterministic functions, $\alpha(\cdot)$ individual fixed effects, and $\varepsilon_t$ transitory shocks.
The realized treatments are determined by the following rules:
\begin{eqnarray} \label{eq:treatment-choice}
    D_{1} &=&\phi _{1}(Y_{0},\xi_0,\eta_1),   \\
    D_{2} &=&\phi _{2}(Y^1,D_{1},\xi_1(D_1),\eta_2), \notag
\end{eqnarray}
where $Y^1\equiv(Y_0,Y_1)\equiv (Y_0,Y_1(D_1))$, $\phi_t(\cdot)$ are deterministic functions, and 
$\xi_0$ and $\xi_1(\cdot)$ are potential individual fixed effects which may evolve over time, with $\xi_1(\cdot)$ possibly subsuming $\xi_0$ as a subvector, and $\eta_t$ are idiosyncratic, time-varying noises. 
(Note $D_{0}=0$ w.p.1, and is suppressed in the notation for $\xi_0,\xi_1(D_1)$.)

Under this general setup, we show a conditional independence between unobserved factors in potential outcomes and those in realized treatment implies sequential exchangeability, but does not warrant parallel trends in general.\bigskip

\begin{proposition} 
\label{thm:evaluating-obs}
Suppose potential outcomes and realized treatments are determined as in (\ref{eq:po-choice}) and (\ref{eq:treatment-choice}), and:
\begin{equation} \label{assn:CI}
    (\alpha(\cdot),\varepsilon_1,\varepsilon_2) \perp (\xi_0,\xi_1(d_1),\eta_1,\eta_2) \,|\, Y_{0} \text{ for all } d_1\text{.}  
\end{equation}
Then sequential exchangeability in Assumption \ref{assn:mse}b holds while $\E[Y_{2}(0,0)-Y_{1}(0)\,|\, D_{1},D_{2}, Y_0]$ in general varies with $(D_{1},D_{2})$.
\end{proposition}
\bigskip 

The condition (\ref{assn:CI}) posits the fixed effects and idiosyncratic errors that influence posteriors are independent of those determining the potential outcomes, after controlling for the initial condition $Y_0$. This condition does \emph{not} further restrict learning. 
This result generalizes to $T\geq 3$, or where the idiosyncratic noises $\varepsilon_t$ and $\eta_t$ are also indexed by contemporary counterfactual treatments $d_t$; we omit these generalizations for brevity.

The general failure of parallel trends in Proposition \ref{thm:evaluating-obs} is intuitive. With forward-looking behavior and learning, the treatment choices $(D_{1},D_{2})$ are a non-trivial function of past observed outcomes $Y_{0},Y_{1}(D_{1})$. 
Hence conditioning on \textit{different} realizations of $(D_1,D_2)$ amounts to conditioning on \textit{different} events on the sample space of $(Y_0,\alpha(\cdot),\varepsilon_1,\xi_0,\xi_1(\cdot),\eta_1,\eta_2)$.
Because $Y_{2}(0,0)-Y_{1}(0)$ is a function of $(Y_{0},\alpha(\cdot),\varepsilon_1,\varepsilon_2)$, its distribution (and mean) generally vary with the conditioning event defined by $(D_1,D_2)$. This non-degeneracy holds in general, even after controlling for the realization in $Y_{0}$.

In \Cref{apx:structural} we show the treatment rule in \eqref{eq:treatment-choice} arises in a structural model, where forward-looking units choose optimal treatments as they learn about their (unknown) fixed effects from past outcomes.

\subsubsection{Example: Gaussian learning of heterogeneous treatment effects} \label{subsec:gaussian}

In this subsection we specialize to a parametric model of Gaussian learning. 
This allows us to derive a closed-form rational treatment rule, and in turn provide sharp characterizations of sequential exchangeability and parallel trends in terms of model primitives such as the information structure and the potential outcome process. 

\noindent \emph{(Potential and observed outcomes.)} Let $T=2$. The potential outcomes are
\begin{equation}\label{eq:ex3-po}
Y_1(d_1) = \gamma Y_0 + \alpha \, d_1+\varepsilon_1 \quad \text{ and } \quad Y_2(d_1,d_2)=\gamma\,Y_1(d_1)+\alpha \, d_2+\varepsilon_2,
\end{equation}
where $ \gamma $ is a \emph{constant} coefficient known to each individual unit, $ \alpha $ is a time-invariant, \textit{heterogeneous} treatment effect (random across units), and $\varepsilon_t$ are stochastic transitory shocks. 
Each unit observes $ Y_1 \equiv Y_1(D_1) $ at $t=1$, and $ Y_2 \equiv Y_2(D_1,D_2) $ at $t=2$, but never observes $ \alpha $ or $ \varepsilon_t $. 
To fix ideas and simplify exposition, let $Y_0$ be degenerate at $0$.\footnote{
    In the notation of \eqref{eq:po-choice} and \eqref{eq:treatment-choice}, $ \alpha(d)=\alpha \, d $, with the treatment effects $\tau_1=\tau_2(d_1)=\alpha$.}

\vspace{.25cm}

\noindent \emph{(Per-period payoffs.)} 
Each unit's payoff in period $t$ is \( u_t(d_t,y_t) + \eta_t(d_t) \), where \[ u_t(d_t,y_t) \equiv y_t - c\,d_t. \]
Each unit knows the constant treatment cost $ c $, observes $\eta_t(0),\eta_t(1)$ in each period $t$, and discounts period-2 payoffs by a known factor $ \rho \in (0,1)$.

\vspace{.25cm}

\noindent \emph{(Information structure.)}
Let $\xi$ denote each unit's heterogeneous signal that is correlated with the unknown treatment effect $\alpha$. 
(Throughout this example, we simplify the unobserved heterogeneity from \Cref{eq:treatment-choice} as $ \xi_0 = \xi_1(\cdot) \equiv \xi $.) 
The exogenous random variables in the model include $(\xi,\alpha)$, transitory outcome shocks $ \varepsilon_t $, and payoff noises $\eta_t(\cdot) \equiv (\eta_t(0),\eta_t(1))$ for $t=1,2$. 

\begin{enumerate}[({I}1)]

    \item The actual joint distribution of $(\xi,\alpha)$ in the data generating process is 
    \[ 
        \begin{pmatrix}\xi\\ \alpha\end{pmatrix}\sim N\!\left(\begin{pmatrix} a_\xi \\ a_\alpha \end{pmatrix},\begin{pmatrix}\sigma_\xi^2 & \sigma_{\xi\alpha} \\  &\sigma_\alpha^2\end{pmatrix}\right). \]
The vectors $(\xi,\alpha)$, $(\varepsilon_1,\varepsilon_2)$, $ \eta_1(\cdot)$, and $\eta_2(\cdot) $ are mutually independent and have finite means, with $ \varepsilon_t\overset{i.i.d.}{\sim}N(0,\sigma_\varepsilon^2) $ and $ \eta_t\equiv\eta_t(1)-\eta_t(0)\overset{i.i.d.}{\sim}N(0,\sigma_\eta^2) $ across units and periods.
Assume $ |\sigma_{\xi\alpha}| < \sigma_\xi\,\sigma_\alpha $ and $\sigma_\xi, \sigma_\alpha, \sigma_\varepsilon,\sigma_\eta > 0 $, which, together with $ a_\xi, a_\alpha $, are constant parameters known to all units.

\item At the start of $t=1$, each unit observes $(\,\xi,\eta_1(\cdot)\,)$; at the start of $t=2$, each unit additionally observes $(\,Y_1,D_1,\eta_2(\cdot)\,)$. The treatment effects $\alpha$ and outcome shocks $\varepsilon_t$ are never observed by the units.
Researchers observe $ (D_1,Y_1,D_2,Y_2) $ but not $ \xi $ in the sample.

\end{enumerate}

\vspace{.25cm}

\noindent\emph{(Units' beliefs.)} A unit's prior for $\alpha$ is the conditional distribution of $\alpha$ given $\xi$, and its posterior at $t=2$ is the conditional distribution of $\alpha$ given its period-2 information. 

Specifically, under (I1)-(I2), a unit's belief at $t=1$ is 
$$ \alpha \mid (\xi, \eta_1(\cdot))  \sim N(m_1(\xi),v_1), $$ where 
\[  m_1(\xi) \equiv a_\alpha + (\tfrac{\sigma_{\xi\alpha}}{\sigma^2_\xi}) \times (\xi - a_\xi)
    \text{ and }  
    v_1 \equiv \sigma_\alpha^2-(\sigma_{\xi\alpha}/\sigma_\xi)^2. \]
Furthermore, given $ \eta_2, Y_1, D_1 $, a unit's posterior at $ t = 2 $ is 
$$ \alpha \mid (\xi,\eta_1(\cdot),\eta_2(\cdot),Y_1,D_1) \sim N(\,m_2(D_1,\xi),v_2(D_1)\,), $$ 
where $ m_2(0,\xi) \equiv  m_1(\xi) $, $ v_2(0) \equiv v_1$, and 
    \[ m_2(1,\xi) \equiv (1-w)\,m_1(\xi)+w\,Y_1 \text{ and } v_2(1) \equiv  w\,\sigma_\varepsilon^2, \]
with \( w \equiv \frac{v_1}{v_1+\sigma_\varepsilon^2} \in (0,1) \).
Note we suppress an argument $Y_1$ in $m_2(1,\xi)$ for simplicity. 

The intuition for the posterior at $t=2$ is the following. 
An untreated unit observes $Y_1 = Y_1(0) =\varepsilon_1$, which is uninformative about $\alpha$; a treated unit observes $Y_1 = Y_1(1) = \alpha+\varepsilon_1$ and applies Gaussian updating with weight $w$. 

\vspace{.25cm}

\noindent\emph{(Rational treatment rule.)} 
We derive the rational treatment rule in this model of Gaussian learning. 
For convenience, we introduce notational shorthands. 
For any $ \sigma > 0 $ and $ m \in \mathbb R $, define $$ \Psi_\sigma(m) \equiv \E(\,\max\{0,\, m + \sigma Z - c \}\,) $$ with $Z\sim N(0,1)$; 
define $\bar\sigma^2\equiv\sigma_\eta^2+w\,v_1$.

\begin{lemma}\label{lem:ex3-rule}
    Given the information structure (I1)-(I2), the rational treatment rules are
    \begin{align}
        D_2 = & \; \ind\{m_2(D_1,\xi) + \eta_2\ge c \}, \label{eq:ex3-D2} \\
        D_1 = & \; \ind\{h(m_1(\xi)) + \eta_1\ge c \}, \label{eq:ex3-D1}
    \end{align}
    where \( h(m)\equiv(1+\rho\gamma)\,m + \rho\,[\Psi_{\bar\sigma}(m)-\Psi_{\sigma_\eta}(m)] \).
\end{lemma}

\vspace{.25cm}

The period-1 index $h(m_1)$ consists of three components: the expected instantaneous gain from choosing $D_1=1$ in period 1 ($m_1$), the discounted dynamic gain due to the impact of a higher period-1 outcome $Y_1$ on the period-2 outcome ($ \rho \gamma m_1 $), and the value of experimentation ($\,\rho\,[\Psi_{\bar\sigma}(m_1)-\Psi_{\sigma_\eta}(m_1)]\,$). 
The last component is the increase in the expected surplus from the optimal choice of period-2 treatment when $D_1=1$ (so that $Y_1$ is informative about $\alpha$).

\vspace{.25cm} 

\noindent\emph{(Evaluating sequential exchangeability vs parallel trends.)}
The next proposition characterizes when sequential exchangeability in \Cref{assn:mse}b and parallel trends in (\ref{assmp:PT}) hold in the current example of Gaussian learning about treatment effects.

\vspace{.25cm}

\begin{proposition}\label{prop:ex3-SEPT}
    Suppose the information structure is (I1)-(I2), and the units follow rational treatment rules as in (\ref{eq:ex3-D2}) and (\ref{eq:ex3-D1}). 
    Then parallel trends holds if and only if $ \gamma = 1 $, while sequential exchangeability holds if and only if $ \sigma_{\xi\alpha} = 0 $.
\end{proposition}

\vspace{.25cm}

\begin{remark}[Violation of parallel trends.] \label{rem:ex3-PT}
    By (\ref{eq:ex3-po}), the trend in untreated outcomes is
    \begin{equation}\label{eq:ex3-trend}
        Y_2(0,0)-Y_1(0)=(\gamma-1)\,\varepsilon_1+\varepsilon_2 .
    \end{equation}
    Since the trend does not involve $\alpha$ and $\varepsilon_2$ is unforeseen, parallel trends can fail only through selection on $\varepsilon_1$. Such selection arises only among units with $D_1=1$, because units with $D_1=0$ observe $Y_1=\varepsilon_1$, which is uninformative about $\alpha$. Let
    \begin{equation}\label{eq:ex3-PTgap}
        \Delta^{PT}\equiv \E[Y_2(0,0)-Y_1(0)\mid D_1=1,D_2=1]-\E[Y_2(0,0)-Y_1(0)\mid D_1=1,D_2=0]
    \end{equation}
    denote the difference in untreated trends between initially treated units that stay treated and those that stop. The proof of \Cref{prop:ex3-SEPT} in \Cref{apx:proofs2} shows that $\Delta^{PT}$ has the sign of $\gamma-1$.
    Intuitively, when $D_1=1$, a higher $\varepsilon_1$ raises $Y_1$ and hence the posterior mean $m_2$, which assigns weight $w>0$ to $Y_1$. By the rational treatment rule in (\ref{eq:ex3-D2}), a higher $m_2$ raises the probability that $D_2=1$, so units that stay treated have higher $\varepsilon_1$ on average. The sign of $\Delta^{PT}$ is therefore that of the coefficient $\gamma-1$ on $\varepsilon_1$ in the untreated trend.
\end{remark}

\vspace{.25cm}

\begin{remark}[Violation of sequential exchangeability.] \label{rem:ex3-SE}
    By (\ref{eq:ex3-D2})-(\ref{eq:ex3-D1}), units choose treatment based on their beliefs $m_1$ and $m_2$, both of which depend on the private signal $\xi$. 
    When $\sigma_{\xi\alpha}\neq0$, $D_2$ varies with $\xi$ given the observed history $(Y_1,D_1)$, and $\xi$ predicts $\alpha$ beyond $(Y_1,D_1)$. This is selection on unobservables: for instance, among units with $D_1=0$, whose $Y_1=\varepsilon_1$ is uninformative about $\alpha$, those with $D_2=1$ have a higher mean of $\alpha$ than those with $D_2=0$.  
    When $\sigma_{\xi\alpha}=0$, $m_1=a_\alpha$ and $m_2$ depends only on $(Y_1,D_1)$. Then $D_1$ depends only on $\eta_1$ and, given $(Y_1,D_1)$, $D_2$ depends only on $\eta_2$; by (I1), both are independent of potential outcomes, also given $(Y_1,D_1)$. While $D_2$ correlates with $\alpha$ through $Y_1$, this is selection on the observed history, on which sequential exchangeability conditions.
\end{remark}

In summary, when the untreated outcomes $Y_1,Y_2(0)$ follow a drift-less random walk ($\gamma = 1$), the parallel trends condition holds. 
Otherwise, learning about heterogeneous treatment effects from one's own outcomes implies that there is systematic selection on the basis of $\varepsilon_1$ in treatment cohorts $D_2\in\{0,1\}$, which violates the parallel trends condition. 
Furthermore, sorting into treatment cohorts on the basis of private information $\xi$ violates sequential exchangeability when such information is correlated with heterogeneous treatment effects.

We further quantify the resulting violations of parallel trends and sequential exchangeability in this Gaussian learning model in \Cref{apx:biases}.
When private signals are uninformative ($\sigma_{\xi\alpha}=0$ so that sequential exchangeability holds), \Cref{prop:ex3-gaps}(a) shows that the trend gap $\Delta^{PT}$ in \eqref{eq:ex3-PTgap} has the sign of $\gamma-1$, and scales with $\sigma_\varepsilon^2$ and the weight in Bayesian updating; the latter is also the weight on $Y_1$ in the period-2 treatment rule. 
When private signals are informative ($\sigma_{\xi\alpha}\neq0$) and $1 + \rho \gamma > \rho / 2$, \Cref{prop:ex3-gaps}(b) shows that units treated in period 1 have larger treatment effects on average, i.e., $\Delta^{SE} \equiv \E[\,Y_1(1)\,\vert\,D_1=1\,] - \E[\,Y_1(1)\,\vert\,D_1=0\,] >0$, which violates sequential exchangeability. 
\Cref{apx:biases} provides closed forms for the gaps $\Delta^{PT} $ and $ \Delta^{SE}$ in special cases.

\section{Outcome Dynamics and Pre-Trend Tests} \label{sec:pretrend}

The parallel trends condition is often justified in practice by conducting a \emph{pre-trend test}, which compares the pre-treatment outcome trends of the treated and untreated groups.
In this section we investigate how dynamics in untreated potential outcomes affect the power of such tests against failures of parallel trends.
We show that, when untreated outcomes follow a dynamic panel data model, the pre-trend null can hold exactly while the subsequent parallel trends restriction fails.
There is then no population pre-trend discrepancy for a test to detect, even though the identifying assumption for the post-treatment comparison (the parallel trends condition) is false. 

\citet{ghanem2026when} studied the conditions under which the pre-trend null and parallel trends hold jointly, uniformly over certain classes of treatment rules.
We consider a different class, namely, rules that are symmetric in the pre-treatment outcomes, and show that within this class the parallel trends condition fails \emph{generically}, even when the pre-trend null holds for \emph{every} rule.
The class we consider does not require knowledge of fixed effects in outcomes, thus allowing forward-looking units to choose treatments rationally while accounting for learning about fixed effects through past outcomes.

Let $T=2$ with no unit treated before the terminal period, i.e., $D_{1}=0$ almost surely. By construction, the initial condition $Y_0$ and the untreated outcome in the first period $ Y_{1}(0) = Y_1 $ are observed. For convenience, write $ D_{2} \equiv D \in\{0,1\} $, and $Y_{2}(0)\equiv Y_{2}(0,0)$. 

The data-generating process for untreated potential outcomes is:
\begin{equation}\label{eq:pretrend-ar1}
    Y_1    = \gamma\,Y_0 + \alpha + \varepsilon_1 \quad\text{and}\quad
    Y_2(0) = \gamma\,Y_1 + \alpha + \varepsilon_2 ,
\end{equation}
with $ \gamma \neq 1 $ being a non-zero, constant autoregressive coefficient, $ \varepsilon_{t} $ time-varying shocks in untreated outcomes and $ \alpha $ individual fixed effects.
Suppose the observed treatment is determined as follows:
\begin{equation} \label{eq:pretrend-rule}
    D\;=\;\ind\big\{\pi_{0}+\pi_{1}Y_{0}+\pi_{2}Y_{1}+\pi_{\xi}'\,\xi>\eta\big\},
\end{equation}
where the $\pi_k$'s are constant coefficients, $\eta$ an idiosyncratic noise, and $\xi$ a vector of unobserved heterogeneity which is arbitrarily correlated with, and may even include, $\alpha$.

\vspace{.25cm}

\begin{assumption}\label{assn:pretrend}
    (a) $(Y_{0},\xi,\alpha)$, $\varepsilon_{1}$, $\varepsilon_{2}$ and $\eta$  are mutually independent.
    (b) $(Y_{0},\xi,\alpha)$ is multivariate normal, $\varepsilon_{1} \sim N(0,\sigma^2_\varepsilon)$ with $\sigma_\varepsilon^2 > 0$, and $\eta\sim N(0,1)$.
\end{assumption}

\vspace{.25cm}

Part (a) allows $Y_{0}$, $\xi$ and $\alpha$ to be arbitrarily dependent; it requires the shocks be independent over time, and of the initial condition and unobserved heterogeneity, and the selection noise $\eta$ be independent of all of them.
Part (b) implies $p\equiv\Pr(D=1)\in(0,1)$.

For any variable $R$, define the selection bias $\Delta_{R}\equiv\E\,(R\,|\,D=1)-\E\,(R\,|\,D=0)$.
The null hypothesis of the pre-trend test and the (unconditional) parallel trends condition are, respectively,
\begin{equation}\label{eq:pretrend-conditions}
    \hpr: \Delta_{Y_{1}(0)-Y_{0}}=0
    \qquad \text{and} \qquad
    \hpt: \Delta_{Y_{2}(0)-Y_{1}(0)}=0.
\end{equation}
Because $\E(R\,|\,D=1)=\E(RD)/p$ and $\E(R\,|\,D=0)=[\E(R)-\E(RD)]/(1-p)$,
\begin{equation}\label{eq:pretrend-delta}
    \Delta_{R}=\frac{\E(RD)-p\,\E(R)}{p(1-p)}=\frac{\mathrm{Cov}(D,R)}{p(1-p)}.
\end{equation}

The next proposition provides sufficient and necessary conditions under which $ \hpr $ and $ \hpt $ hold respectively.
For convenience, denote the selection index by $ \widetilde Z \equiv \pi_{1}Y_{0}+\pi_{2}Y_{1} + \pi_{\xi}'\xi $; let $ f^*(\cdot) $ be the density of $\pi_0 + \widetilde Z - \eta$, with $f^*(0)>0$. 
Define $$ A(\gamma,\pi,\bsig) \equiv \mathrm{Cov}(\widetilde Z,Y_1-Y_0), $$ which is determined by $\gamma$ and $\pi\equiv(\pi_1,\pi_2,\pi_\xi')$, as well as variances and covariances of $ Y_0,\alpha,\xi,\varepsilon_1 $, collected in the notation $ \bsig \equiv (\sigma_{0}^{2}, \sigma_{\alpha}^{2}, \sigma_{\varepsilon}^{2}, \sigma_{0\alpha}, \sigma_{0\xi}', \sigma_{\xi\alpha}') $, with $\sigma_0^2,\sigma_{0\alpha},\sigma_{0\xi}$ being the variance and covariances associated with $Y_0$.

\vspace{.25cm}

\begin{proposition}\label{prop:pretrend}
    Under \Cref{assn:pretrend}, the null in pre-trend tests holds if and only if $A(\gamma,\pi,\bsig)=0$, and the parallel trends condition holds if and only if $ A(\gamma,\pi,\bsig) = \pi_{2}\,\sigma_{\varepsilon}^{2}/\gamma$.
\end{proposition}

\vspace{.25cm}

With $f^*(0)>0$, the proposition implies $\mathrm H^{\,0}_{\mathrm{PRE}}$ holds while $ \mathrm H^{\,0}_{\mathrm{PT}}$ fails simultaneously if and only if $$ A(\gamma,\pi,\bsig)=0 \text{ and } \pi_{2} \neq 0, $$ in which case the difference-in-differences estimand deviates from the target ATT by a bias term $\Delta_{Y_{2}(0)-Y_{1}} = -\tfrac{\pi_{2}\,\sigma_{\varepsilon}^{2}\,f^*(0)}{p(1-p)}$.
Since $\sigma_\varepsilon^2 > 0$, the two restrictions cannot both hold when $\pi_2 \neq 0$.
When the pre-trend null holds, the bias in the difference-in-differences estimator is non-zero when $\pi_2\neq0$, with its sign given by $-\pi_{2}$. 

\vspace{.25cm}

\begin{remark}
    If the treatment rule is symmetric in its most recent untreated outcomes (i.e., $\pi_{1}=\pi_{2}$), then $A(\gamma,\pi,\bsig)$ reduces to $\pi_{1}[\mathrm{Var}(Y_{1})-\mathrm{Var}(Y_{0})]+\pi_{\xi}'[\mathrm{Cov}(\xi,Y_{1})-\mathrm{Cov}(\xi,Y_{0})]$.
    A sufficient condition for the null in pre-trend tests to hold uniformly over such a class of treatment rules is $\mathrm{Var}(Y_{0})=\mathrm{Var}(Y_{1})$ and $\mathrm{Cov}(\xi,Y_{0})=\mathrm{Cov}(\xi,Y_{1})$, which holds if the joint distribution of $(Y_{0},Y_{1},\xi)$ is exchangeable in its first two components.
\end{remark}

\vspace{.25cm}

\begin{example}[Treatment determined by past outcomes]
Consider an untreated process as in \eqref{eq:pretrend-ar1}:
\[
Y_1=\gamma Y_0+\alpha+\varepsilon_1,
\qquad
Y_2(0)=\gamma Y_1+\alpha+\varepsilon_2,
\]
with \(|\gamma|\in(0,1)\). Suppose the observed treatment is assigned according to the sum of recent performances with no persistent unobserved heterogeneity $\xi$,
\[ D=1\{\pi_0+Y_0+Y_1 > \eta \}, \]
where \(\eta\) is independent assignment noise. Suppose \Cref{assn:pretrend} holds with \(\operatorname{Var}(Y_1)=\operatorname{Var}(Y_0)\), which is implied when $\sigma^2_0 = (2\gamma\sigma_{0\alpha} + \sigma^2_\alpha + \sigma^2_\varepsilon)/(1-\gamma^2)$. 
Then
\[
A(\gamma,\pi,\Sigma)=\operatorname{Cov}(Y_0+Y_1,Y_1-Y_0)
=\operatorname{Var}(Y_1)-\operatorname{Var}(Y_0)=0.
\]
Thus the pre-trend null holds. However,
\[
Y_2(0)-Y_1=\gamma(Y_1-Y_0)+\varepsilon_2-\varepsilon_1,
\]
and treatment depends on \(Y_1\), which contains \(\varepsilon_1\). Consequently,
\[
\operatorname{Cov}\bigl(D,Y_2(0)-Y_1\bigr) = f^* (0) \operatorname{Cov}(Y_0+Y_1,Y_2(0)-Y_1)
=-\sigma_{\varepsilon}^{2}f^*(0)\neq 0,
\]
where $f^*$ is the density of $\pi_0+Y_1+Y_0-\eta$.\footnote{ 
    The first equality holds by proof of \Cref{prop:pretrend}; the second by $ \operatorname{Cov}(Y_0+Y_1,Y_1-Y_0) =0$ above.} 
So, parallel trends fails. Economically, treatment is assigned to units with high recent outcomes, but part of the high value of \(Y_1\) reflects a transitory shock that does not persist, generating regression to the mean despite equal pre-treatment trends. \hfill $\square$
\end{example}

\section{Conclusion} \label{sec:conclusion}

We studied dynamic panel data methods in a causal framework where potential outcomes depend on past treatments. 
We derived sufficient, and in some senses necessary, conditions for instrumented first-difference estimators to be causally interpretable.
A motivation and underlying assumption for one set of these results was sequential exchangeability. 
Besides investigating the plausibility of this assumption in both experimental and observational settings (with a structural model of dynamic treatment choices), we showed how to conduct causal inference using dynamic panel data estimators when this assumption fails. 
A natural avenue for future work is to consider (partial) identification when sequential exchangeability is relaxed or holds only conditional on a set of unobservables. 

\newpage

\appendix

% Equations in the appendix are numbered (A1), (A2), ..., (B1), (B2), ..., restarting in each appendix section
\numberwithin{equation}{section}
\renewcommand{\theequation}{\thesection\arabic{equation}}
\renewcommand{\theHequation}{\thesection.\arabic{equation}}

% Appendix-specific numbering of results (e.g., Lemma C1, Proposition D1)
\newcommand{\appendixnumbering}[1]{%
  \setcounter{assumption}{0}\setcounter{observation}{0}%
  \setcounter{proposition}{0}\setcounter{lemma}{0}%
  \renewcommand{\theassumption}{#1.\arabic{assumption}}\def\theHassumption{#1.\arabic{assumption}}%
  \renewcommand{\theobservation}{#1.\arabic{observation}}\def\theHobservation{#1.\arabic{observation}}%
  \renewcommand{\theproposition}{#1.\arabic{proposition}}\def\theHproposition{#1.\arabic{proposition}}%
  \renewcommand{\thelemma}{#1.\arabic{lemma}}\def\theHlemma{#1.\arabic{lemma}}%
}

\section{Proofs in \texorpdfstring{\Cref{sec:framework,sec:DPR}}{Sections 2 and 3}}  \label{apx:proofs}

\begin{proof}[Proof of Proposition \ref{thm:mse-trend-equivalence}]
\noindent
We fix $s \geq 1$ throughout the following proof. For $t = s$, the equivalence is immediate:
\begin{align*}
&
\E[ Y_s (d^s) \mid  Y^{s-1}, D^{s-1} = d^{s-1}, D_s = d_s'] \quad \text{is invariant in $d_s'$} \\
\iff &
\E[ Y_s (d^s) \mid  Y^{s-1} (d^{s-1}), D^{s-1} = d^{s-1}, D_s = d_s'] \quad \text{is invariant in $d_s'$} \\
\iff &
\E[ Y_s (d^s) \mid  Y^{s-1} (d^{s-1}), D^{s-1} = d^{s-1}, D_s = d_s'] - Y_{s-1} (d^{s-1}) \quad \text{is invariant in $d_s'$} \\
\iff &
\E[ Y_s (d^s) - Y_{s-1} (d^{s-1}) \mid  Y^{s-1} (d^{s-1}), D^{s-1} = d^{s-1}, D_s = d_s'] \quad \text{is invariant in $d_s'$.}
\end{align*}

\noindent
For $t > s$, Assumption \ref{assn:mse}b implies: 
\begin{align*}
& 
\E[ Y_t (d^t) \mid  Y^{s-1}, D^{s-1} = d^{s-1}, D_s = d_s'] \quad \text{is invariant in $d_s'$} \\
& 
\E[ Y_{t-1} (d^{t-1}) \mid  Y^{s-1}, D^{s-1} = d^{s-1}, D_s = d_s'] \quad \text{is invariant in $d_s'$} 
\end{align*}
and therefore: 
\[
\E[ Y_t (d^t) - Y_{t-1} (d^{t-1}) \mid  Y^{s-1}, D^{s-1} = d^{s-1}, D_s = d_s'] \quad \text{is invariant in $d_s'$}.
\]

\noindent
We prove the reverse direction of implication for $t > s$ inductively.
Suppose we have shown that, for some $t \geq s$, the conditions on the trends $Y_t(d^{t})-Y_{t-1}(d^{t-1})$ in Proposition \ref{thm:mse-trend-equivalence} imply:
\[
    \E[ Y_t (d^t) \mid  Y^{s-1}, D^{s-1} = d^{s-1}, D_s = d_s'] \quad \text{is invariant in $d_s'$}. 
\]
By the conditions on the trends, we also have: 
\[
    \E[ Y_{t+1} (d^{t+1}) - Y_t (d^t) \mid  Y^{s-1}, D^{s-1} = d^{s-1}, D_s = d_s'] \quad \text{is invariant in $d_s'$.} 
\]
Combining these two equalities then implies: 
\[
    \E[ Y_{t+1} (d^{t+1}) \mid  Y^{s-1}, D^{s-1} = d^{s-1}, D_s = d_s'] \quad \text{is invariant in $d_s'$.} 
\]
Finally, recall we already showed the equivalence (between the two sets of conditions on the levels and the trends of potential outcomes) for $t = s$ at the start of this proof. Hence by induction the claim of the reverse direction of implication for $t>s$ holds.
\end{proof}

\bigskip

\begin{proof}[Proof of Lemma \ref{thm:iv-fwl}] 
As in \eqref{eq:fd-iv}, let $\Delta X \equiv (\Delta D_2, \Delta Y_1, 1) $ denote the (row) vector of covariates with $T = 2$ and $Z \equiv (\E( \Delta D_2 \mid  Y_0, D_1), \E( \Delta Y_1\mid  Y_0, D_1), 1)$ denote the vector of instruments. 
To reiterate, the first-differenced IV estimand $(\widetilde \beta, \widetilde \gamma, \widetilde{\Delta \theta_2})'$ in \eqref{eq:fd-iv} is: 
\[
    \E( Z' \Delta X)^{-1} \E( Z' \Delta Y_2)
\]
Note that $Z$ contains conditional expectations of each element of $\Delta X$. By the law of iterated expectations, this estimand is equal to an OLS estimand: 
\[
    \E( Z' Z)^{-1} \E( Z' \Delta Y_2).
\]
By the formula for partitioned regressions, a.k.a.~a population version of the Frisch-Waugh-Lovell Theorem, the estimand for the slope coefficient of the first component in $Z$ above is equal to 
\[
\frac{\E(W \Delta Y_2 )}{\E( W^2)},
\]
where $W$, defined in \eqref{eq:weights}, is the projection error from the linear projection of $\E(\Delta D_{2}\mid Y_{0},D_{1})$ onto $\E(\Delta Y_{1}\mid Y_{0},D_{1})$ and a constant intercept.
Plugging in the decomposition \eqref{eq:Y2-decompose} of $\Delta Y_2$ yields the desired expression in \eqref{eq:iv-fwl}. 
\end{proof}

\ \ \

\begin{proof}[Proof of Proposition \ref{thm:lh1}]
For the proof, we collect the following two properties of a distribution $F \in \mathcal{F}$ satisfying Assumptions \ref{assn:po-treatments}, \ref{assn:mse}b, \ref{assn:iv-rank}. 
First, by \eqref{eq:Y1-decompose} and Assumption \ref{assn:mse}b, 
\begin{equation}
\label{eq:Delta-Y1}
    \E [ \Delta Y_1 \mid Y_0, D_1]
    = 
    \E [ \delta_1 + D_1 \tau_1 \mid Y_0, D_1 ] 
    = 
    \E [ \delta_1 \mid Y_0] + 
    D_1 \E [ \tau_1 \mid Y_0].
\end{equation}
Second, for any (extended) distribution satisfying Assumptions \ref{assn:po-treatments} and \ref{assn:iv-rank}, Lemma \ref{thm:iv-fwl} gives
\begin{equation}
\label{eq:beta}
    \widetilde{\beta} 
    = 
    \frac{ \E [ W \delta_2 (D_1)] + \E [ W D_2 \tau_2 (D_1)]}{\E (W^2) }.
\end{equation}
Next, we prove the following statements: $1 \implies 2 \implies 3 \implies 1$.

$(\mathbf{1 \implies 2})$. Under Assumption \ref{assn:limited-heterogeneity} (a) and (b) and from \eqref{eq:Delta-Y1},
\[
\E [ \Delta Y_1 | Y_0, D_1 ] = \E [ \delta_1] + D_1 \E [ \tau_1],
\]
with $\E [ \tau_1] \neq 0$ by Assumption \ref{assn:iv-rank}.
Projecting $ \E[ \Delta D_2 \mid Y_0, D_1]$ on $(1,\E [ \Delta Y_1 \mid Y_0, D_1])$ is therefore identical to conditioning its mean on $D_1$, giving:
\begin{align}
W 
&= 
\E[\Delta D_{2}\mid Y_{0},D_{1}]-\E[\E[\Delta
D_{2}\mid Y_{0},D_{1}]\mid D_{1}] 
\notag \\
&=
\E[D_{2}\mid Y_{0},D_{1}]-\E[D_{2}\mid D_{1}].
\label{eq:projection-error}
\end{align}
Consequently, $\E [ W \mid D_1 ] = 0$ and
\[
\E [ W D_2 ] = 
\E [ W \{ W + \E [ D_2 \mid D_1 ] \}] = \E [ W^2],
\]
which establishes that the weights in \eqref{eq:beta-no-contam} average to one. 

To recover the expression for $\widetilde{\beta}$ in \eqref{eq:beta-no-contam}, by \eqref{eq:beta} it remains to show that $E [ W \delta_2 (D_1)] = 0$.
For each $d \in \{0,1\}$, Assumption \ref{assn:mse}b and Assumption \ref{assn:limited-heterogeneity} (c) give:
\[
\E [ \delta_2 (D_1) \mid Y_0, D_1 = d] = \E [ \delta_2 (d) \mid Y_0] = \E [ \delta_2 (d)]
\]
The conditional mean therefore depends only on $D_1$, and so:
\[
\E [ \delta_2 (D_1) \mid Y_0, D_1] = \E [ \delta_2 (D_1) \mid D_1]
\]
Since $W$ is a function of $(Y_0, D_1)$, iterated expectations yield: 
\begin{align*}
    \E [ W \delta_2 (D_1)] 
    &= \E [ W \E [ \delta_2 (D_1) \mid Y_0, D_1]] \\
    &= \E [ W \E [ \delta_2 (D_1) \mid D_1]] \\
    &= \E [ \E [ W \mid D_1] E [ \delta_2 (D_1) \mid D_1]] \\
    &
    = 0
\end{align*}
where the last equality uses $E[W \mid D_1 ] = 0$. 

$(\mathbf{2 \implies 3})$. If $\tau_2 (D_1) \equiv \bar{\tau} > 0$, statement 2 gives $\widetilde{\beta} = \bar{\tau} > 0$. 

$(\mathbf{3 \implies 1})$. We prove the contrapositive: fix a marginal distribution $G \in \mathcal{F}_V$ as specified in the proposition and suppose that (a), (b), or (c) fails. 
We construct an extended distribution $F \in \mathcal{F}$ with $F_V = G$ and with positive constant effects $\bar{\tau} > 0$, and yet $\widetilde{\beta} < 0$. 

We proceed with two constructions. In each, we extend the fixed distribution $G$ by setting $Y_2 (d,1) = Y_2 (d,0) + \bar{\tau}$ for $d=0,1$, and generate $D_2$ using randomness independent of the fixed variables, with the indicated conditional probability. 
Constant shifts and the conditional randomization preserve Assumptions \ref{assn:po-treatments}, \ref{assn:mse}b and finite moments. 

Consider first the case where the fixed distribution violates Assumption \ref{assn:limited-heterogeneity} (a) or (b).
Within the proof, define $R$ to be the residual from projecting $D_1$ on $(1,E[\Delta Y_1 \mid Y_0, D_1])$, so: 
\[
D_1 = a + b E [ \Delta Y_1 \mid Y_0, D_1] + R 
\]
for some constants $a,b$. 
It must be that $E[R^2] > 0$. For the sake of contradiction, suppose not: $R = 0$ almost surely, i.e.,
\[
D_1 = a + b E [ \Delta Y_1 \mid Y_0, D_1] \quad \text{a.s.}
\]
Furthermore $b \neq 0$, since otherwise $D_1$ is essentially constant, contradicting overlap.
Substituting from \eqref{eq:Delta-Y1},
\[
D_1 = a + b 
\{
E [ \delta_1 \mid Y_0] + D_1 E [ \tau_1 \mid Y_0]
\}
\quad \text{a.s.}
\]
By overlap, this must hold for each realization of $D_1$:
\[
\begin{array}{rl}
    0 &= a + b E [ \delta_1 \mid Y_0] \\
    1 &= a + b \{ 
    E [ \delta_1 \mid Y_0] + 
    E [ \tau_1 \mid Y_0]
    \}
\end{array}
\quad \text{a.s.}
\]
Solving gives $E [ \delta_1 \mid Y_0] = -a/b$ and $E [ \tau_1 \mid Y_0] = 1/b$ almost surely. But then Assumption \ref{assn:limited-heterogeneity} (a) and (b) would both hold, contradicting their failure. Hence $E[R^2] > 0$.

Next, choose the conditionally random period-2 assignment rule: 
\[
P(D_2 = 1 \mid Y_0, D_1 ) = \frac{1}{4} + \frac{1}{2} D_1.
\]
Then: 
\begin{align*}
E[\Delta D_2 \mid Y_0, D_1] 
= 
\frac{1}{4} - \frac{1}{2} D_1
= 
\left(
    \frac{1}{4} - \frac{a}{2}
\right)
- \frac{b}{2} E [ \Delta Y_1 \mid Y_0, D_1] - \frac{1}{2} R
\end{align*}
where the second equality plugs in for the linear projection of $D_1$ on $(1,E[\Delta Y_1 \mid Y_0, D_1])$.
By construction, $R$ is orthogonal to $(1,E[\Delta Y_1 \mid Y_0, D_1])$. The first two terms therefore give the fitted value, with residual $W = - R/2$, and so $E[W^2] = E[R^2]/4$.
Also, 
\begin{align*}
    E [ W D_2 ] 
    =
    E [W E [ D_2 \mid Y_0, D_1]] 
    = 
    E \left[ W \left( \frac{1}{4} + \frac{1}{2} D_1 \right) \right]
    =
    \frac{1}{2} E [ W D_1]
    = 
    -\frac{1}{4} E[ R D_1]
\end{align*}
where the first equality follows by the LIE, the second by definition of the assignment rule for $D_2$, the third because $E[W] = 0$ by definition as a projection error, and the last because $W = -R/2$.
Since $D_1$ equals its fitted value plus $R$, and $R$ is orthogonal to that fitted value, $E [R D_1] = E [ R^2]$. 
Combining with the previous equalities yields $E [ W D_2] = - E [ R^2]/4$.

For a constant $\bar{\tau}$, plugging the preceding equalities into \eqref{eq:beta} yields:
\[
\widetilde{\beta} = -\frac{2 E [ R \delta_2 (D_1)]}{E[R^2]} - \bar{\tau}. 
\]
A marginal distribution $G \in \mathcal{F}_V$ determines $\delta_2 (D_1)$ and $R$, and consequently the first term. 
We can then take $F$ with $F_V = G$ and with $\bar{\tau} > 0$ sufficiently large so that $\widetilde{\beta} < 0$. 

We consider the remaining case where the fixed distribution $G \in \mathcal{F}_V$ satisfies Assumption \ref{assn:limited-heterogeneity} (a) and (b), but violates (c). 
Given the partial distribution $(Y_0, D_1, Y_1 (\cdot), Y_2 (\cdot, 0))$, define within this proof:
\begin{equation}
\label{eq:H}
H \equiv E [ \delta_2 (D_1) \mid Y_0, D_1] - E [ \delta_2 (D_1) \mid D_1]
\end{equation}
Assumption \ref{assn:mse}b and overlap imply $P (H \neq 0) > 0$; otherwise (c) would hold. 

Choose the conditionally random period-2 assignment rule: 
\[
P(D_2 = 1 \mid Y_0, D_1 ) 
= 
\frac{1}{2} - \frac{1}{4} \sgn (H). 
\]
The projection error formula \eqref{eq:projection-error} established under (a) and (b) gives: 
\[
W = - \frac{1}{4} \{ 
\sgn (H) - E [ \sgn (H) \mid D_1] \},
\quad
E [ W \mid D_1 ] = 0, 
\quad 
E [ W D_2 ] = E [ W^2],
\]
so: 
\[
E [ W \delta_2 (D_1)] 
= E [ W H] 
= -\frac{1}{4} E [\,|H|\,] < 0
\]
where the first equality follows by the law of iterated expectation conditioning on $(Y_0,D_1)$, plugging in for $E[ \delta_2 (D_1) \mid Y_0, D_1]$ from \eqref{eq:H}, and the fact that $E[W \mid D_1] = 0$; the second equality follows from plugging in for $W$ and the fact that $E[ H \mid D_1] = 0$ by definition \eqref{eq:H} of $H$; and the strict inequality follows because $P (H \neq 0) > 0$.
It also follows that $E [W^2] > 0$, since otherwise $E [ W \delta_2 (D_1)] =0$.
Collecting preceding arguments with a constant $\bar{\tau}$, \eqref{eq:beta}  becomes:
\[
\widetilde{\beta} =  - 
\frac{E[ |H| ]}{4 E[W^2]} 
+ \bar{\tau}
\]
A given marginal distribution $G \in \mathcal{F}_V$ determines $H$ and $W$, and consequently the first, strictly negative term. 
We can then take $F$ with $F_V = G$ and with $\bar{\tau} > 0$ sufficiently small so that $\widetilde{\beta} < 0$. 

Finally, note that in both constructions, $E[ \Delta Y_1 \mid Y_0, D_1] $ remains nonconstant by the original rank condition, and $E[W^2] > 0$ as established above. 
Hence the extended distributions also satisfy Assumption \ref{assn:iv-rank}.
Each construction therefore contradicts statement 3, completing the proof. 
\end{proof}

\bigskip

\begin{proof}[Proof of Proposition \ref{thm:convex-te}]
Note that Assumption \ref{assn:mse}a implies \ref{assn:mse}b; we invoke the latter when it suffices.
Define:
\[
M \equiv E [ \tau_2 (D_1) \mid Y_0, Y_1, D_1] - E [ \tau_2 (D_1) \mid D_1],
\]
which implies $E[ M \mid D_1] = 0$.
Our maintained assumptions imply two useful properties of $M$. First, Assumption \ref{assn:mse}a implies:
\[
E [ \tau_2 (D_1) \mid Y_0, Y_1, D_1 = d]
= E [ \tau_2 (d) \mid Y_0, Y_1 (d)]
\]
Under condition (d), this simplifies to $E [ \tau_2 (d)] $, which in turn equals $E [ \tau_2 (D_1) \mid D_1 = d] = E[E[\tau_2(d) \mid Y_0,Y_1(d), D_1 = d] \mid D_1=d]  = E[\tau_2(d)] $. Thus, condition (d) implies $M = 0$ a.s.
Conversely, also invoking overlap, $P(D_1 = d \mid Y_0, Y_1 (d)) = P(D_1 = d \mid Y_0) > 0$, and so $M = 0$ a.s.~implies condition (d).
Summarizing, under the maintained assumptions, condition (d) is equivalent to $M = 0$ a.s. 

Second, by Proposition \ref{thm:lh1}, $C_\delta = 0$ in Lemma \ref{thm:iv-fwl}; for the remaining bias term,
\begin{align}
    C_\tau 
    & \equiv 
    E [ W (D_2 - W) \tau_2 (D_1)] \notag \\
    &= 
    E [ W (D_2 - W) E [ \tau_2 (D_1) \mid Y_0, Y_1, D_1, D_2]] \notag \\
    &= 
    E [ W (D_2 - W) E [ \tau_2 (D_1) \mid Y_0, Y_1, D_1]] \notag \\
    &=
    E [ W (D_2 - W) \{ M + E [ \tau_2 (D_1) \mid D_1 ] \}] \notag \\
    &= 
    E [ W (D_2 - W) M ] \label{eq:C-M}
\end{align}
where the first equality recalls the definition, the second is by LIE, the third is by Assumption \ref{assn:mse}b,
the fourth by definition of $M$, and the fifth because $E [ W (D_2 - W) \mid D_1] = 0$ implied by \eqref{eq:weights-convex}.

$\mathbf{(1 \implies 2)}$: It suffices to show that $C_\tau = 0$. This follows from the preceding two points, namely the expression \eqref{eq:C-M} of $C_\tau$ in terms of $M$ and the fact that (d) is equivalent to $M = 0$ a.s.

$\mathbf{(1 \implies 3)}$: Under Assumption \ref{assn:mse}a and Assumption \ref{assn:limited-heterogeneity} (d),
\[
E [ \tau_2 (D_1) \mid Y_0, D_1 = d] = E [ \tau_2 (d) ], 
\quad d = 0,1.
\]
Since $W$ is a function of $(Y_0, D_1)$, the representation \eqref{eq:te-convex} established in $(1 \implies 2)$ gives: 
\[
\widetilde{\beta} = \sum_{d=0}^{1} 
\frac{E [ W^2 1 \{ D_1 = d \}]}{E[W^2]} E [\tau_2 (d)]
\]
Thus $\widetilde{\beta}$ is a convex combination of $E[\tau_2 (0)]$ and $E[\tau_2 (1)]$. Taking $c_G = \max_{d=0,1} \, | E [\tau_2 (d)] |$ yields the desired bound.

$\mathbf{(2 \implies 1 \text{ and } 3 \implies 1)}$: We prove both contrapositives together. Suppose (d) fails, so that $P (M \neq 0) > 0$. 
Keeping the distribution $G$ fixed, we construct period-2 assignment rules indexed by $0 < | \lambda | \leq 1/8$. 
For each constructed distribution, $C_\tau \neq 0$, contradicting statement 2. 
We further show that the resulting coefficients are unbounded, contradicting statement 3.

In each construction below, generate $D_2$ using independent randomness with the stated conditional probability, which preserves Assumption \ref{assn:mse}a. 
Within this proof, define: 
\[
H = E [ M \mid Y_0, D_1]
\]
implying $E [ H \mid D_1] = 0$. We distinguish two cases: heterogeneity is already present in this conditional mean, or it appears only after additionally conditioning on $Y_1$. 

In the first case, $P(H \neq 0) > 0$. Conditionally randomize $D_2$ with probabilities:
\[
P (D_2 = 1 \mid Y_0, Y_1, D_1) = \frac{1}{2} + \lambda
\{ \sgn (H) - E [ \sgn (H) \mid D_1] \}
\]
This probability also equals $ E [ D_2 \mid Y_0, D_1]$ since it does not vary with $Y_1$. 
Additionally it implies $E [ D_2 \mid D_1] = 1/2.$
Plugging into the simplified weight formula \eqref{eq:weights-convex} gives: 
\[
W = \lambda \{ \sgn (H) - E [ \sgn(H) \mid D_1] \}
\]
and therefore also $E [ D_2 - W \mid Y_0, Y_1, D_1] = 1/2$.
Then:
\begin{align*}
    C_\tau 
    &= E[ W (D_2 - W) M] \\
    &= E [ E [ D_2 - W \mid Y_0, Y_1, D_1] W M] \\
    &= \frac{1}{2} E [ W M] = \frac{1}{2} E [ W H] = \frac{\lambda}{2} E [ | H |] \neq 0
\end{align*}
where the first equality recalls \eqref{eq:C-M}, the second is by LIE conditioning on $(Y_0, Y_1, D_1)$, the third plugs in, the fourth is by LIE conditioning on $(Y_0, D_1)$, and the fifth plugs in for $W$ and uses $E [ H \mid D_1] = 0$. 
Additionally, from the expression for $W$,
\[
E [ W^2] = \lambda^2 E [ \{ \sgn (H) - E [ \sgn(H) \mid D_1] \}^2 ]
\]
In the second case, $H = E [ M \mid Y_0, D_1] = 0$ a.s., while still $P( M \neq 0) > 0$. We construct an assignment rule that responds to the additional information in $Y_1$, while retaining variation in the first stage.

The rank condition in Assumption \ref{assn:iv-rank} requires that $Y_0$ be non-constant; otherwise, under Assumption \ref{assn:limited-heterogeneity} (a) and (b), both first-stage conditional means $E [ \Delta D_2 \mid Y_0, D_1] $ and $E[ \Delta Y_1 \mid Y_0, D_1]$ would be functions of only binary $D_1$, giving $W = 0$. Choose a cutoff $k$ such that $0 < P(Y_0 \leq k) < 1$ and define: 
\[
K = 1 \{ Y_0 \leq k \} - P (Y_0 \leq k \mid D_1)
\]
implying $E[K \mid D_1] = 0$. Initial-treatment overlap implies $0 < P(Y_0 \leq k \mid D_1 = d) < 1$ for $d = 0,1$, so $K \neq 0$ a.s.

Conditionally randomize $D_2$ with probabilities:
\[
P (D_2 = 1 \mid Y_0, Y_1, D_1) = \frac{1}{2} + \lambda K + \frac{K}{8} \{ \sgn (M) - E [ \sgn (M) \mid Y_0, D_1] \}
\]
The full expression again equals $E [ D_2 \mid Y_0, Y_1, D_1]$. The last term averages to zero conditional on $(Y_0, D_1)$. Consequently, 
\[
E [ D_2 \mid Y_0, D_1] = \frac{1}{2} + \lambda K, 
\qquad E [ D_2 \mid D_1 ] = \frac{1}{2}, 
\qquad W = \lambda K
\]
Then: 
\begin{align*}
    C_\tau 
    &= 
    E [ E [ D_2 - W \mid Y_0, Y_1, D_1] WM ] \\
    &= 
    \frac{\lambda}{2} E [ KM] 
    + \frac{\lambda}{8} E [ K^2 M \{ \sgn (M) - E [\sgn (M) \mid Y_0, D_1] \} ] \\
    &= 
    \frac{\lambda}{8} E [ K^2 | M | ] \neq 0
\end{align*}
where the first equality is by \eqref{eq:C-M} and LIE conditioning on $(Y_0, Y_1, D_1)$, the second plugs in from the preceding expressions, and the third simplifies using $E [ M \mid Y_0, D_1]= 0$, which in turn implies that $E [ KM ] = 0$ and $E [K^2 M E [ \sgn (M) \mid Y_0, D_1]] = 0$. 
The expression is nonzero because $\lambda \neq 0$, $K^2 > 0$ a.s., and $P(M \neq 0) > 0$.
From the preceding expression for $W$, $E [ W^2] = \lambda^2  E [ K^2]$. 

In both constructions, 
the assignment probabilities lie in $[1/8, 7/8]$, and the fixed potential outcomes retain their finite moments. 
Since $C_\tau \neq 0$, we have $E[ W^2] > 0$. 
Moreover, $E[\Delta Y_1 \mid Y_0, D_1]$ remains non-constant by the original rank condition (Assumption \ref{assn:iv-rank}). Thus every constructed distribution belongs to $\mathcal{F}^*$ and has marginal $G$. 

Because $C_\delta = 0$ and $C_\tau \neq 0$, each construction violates statement 2. 
To establish failure of statement 3, write: 
\[
\widetilde{\beta} = E 
\left[ \frac{W^2}{E[W^2]} \tau_2 (D_1)
\right]
+ \frac{C_\tau}{E[W^2]}
\]
In both cases, $W/\lambda$ is a fixed, bounded function of $(Y_0,D_1)$. The convex-weight term is therefore finite and independent of $\lambda$. Meanwhile, the preceding calculations show that $C_\tau / E[W^2]$ equals a strictly positive constant divided by $\lambda$. Hence $\widetilde{\beta} \to +\infty$ as $\lambda \downarrow 0$ and $\widetilde{\beta} \to - \infty$ as $\lambda \uparrow 0$. 
Thus statement 3 also fails.
\end{proof}

\bigskip

\begin{proof}[Proof of Proposition  \ref{thm:IV-arupo}]
As in the text, we suppress the realized value of initial condition $y_0$ in the notation of the conditioning event $\mathcal S_{t-1}(y_0)$, the instrument $Z_{it}(y_0)$, and the target parameter $\mathrm{ATFT}_t(y_0)$. 

Under Assumption \ref{assn:arupo}, we use \eqref{eq:fd-obs} to write the mean of $\Delta Y_{it}$ given $ H^*_{it-1}$ and $\mathcal S_{t-1}$ as:
\begin{equation}
\label{eq:arupo-cond-mean} 
\E(\Delta Y_{it}\vert H^*_{it-1},\mathcal S_{t-1}) = \E(\taus\,D_{it}\vert H^*_{it-1},\mathcal S_{t-1}) + \gamma\E(\Delta Y_{it-1}\vert H^*_{it-1},\mathcal S_{t-1}) + \Delta\theta_t,
\end{equation}
because, on $\mathcal S_{t-1}$, $D_{it-1}=0$ and $\tau_{it}=\taus$, while $\E(\Delta\varepsilon_{it}\vert H_{it-1})=0$ by \eqref{eq:seq-exo} and both $H^*_{it-1}$ and the event $\mathcal S_{t-1}$ are functions of $H_{it-1}$. Moreover, 
\begin{eqnarray*}
    \E(\taus\,D_{it}\vert H^*_{it-1},\mathcal S_{t-1}) 
    & = & \E(\taus\,\vert\, D_{it}=1,H^*_{it-1},\mathcal S_{t-1})\times \E(D_{it}\vert H^*_{it-1},\mathcal S_{t-1}) \\
    & = & \underset{\mathrm{ATFT}_t}{\underbrace{\E(\taus\,\vert\, D_{it}=1,\mathcal S_{t-1})}}\times \E(\Delta D_{it}\vert H^*_{it-1},\mathcal S_{t-1}),
\end{eqnarray*}
where the first equality uses $\taus D_{it}=0$ when $D_{it}=0$, and the second equality holds under \Cref{assn:MI-hte} because $H^*_{it-1}$ is a function of $\{ \varepsilon_i^{t-2}(\cdot),\alpha_i \}$ on $\mathcal S_{t-1}(y_0)$, and because $D_{it}=\Delta D_{it}$ on $\mathcal S_{t-1}$.
This implies the mean of $\Delta Y_{it}$ conditional on $H^*_{it-1}$ and $\mathcal S_{t-1}$ is linear in $Z_{it}$, with the coefficients and intercept being \( (\mathrm{ATFT}_t,\gamma,\Delta\theta_t)\). Hence, 
\begin{equation*}
    (\,\mathrm{ATFT}_t,\,\gamma,\,\Delta\theta_t\,)' =  
    \E(Z_{it}'Z_{it}\vert\mathcal S_{t-1})^{-1}\E(Z_{it}'\Delta Y_{it}\vert \mathcal S_{t-1}) 
    = \E(Z_{it}'X_{it}\vert\mathcal S_{t-1})^{-1}\E(Z_{it}'\Delta Y_{it}\vert \mathcal S_{t-1}),     
\end{equation*}
where \(\E(Z_{it}'X_{it}\vert\mathcal S_{t-1})=\E(Z_{it}'Z_{it}\vert\mathcal S_{t-1})\) by the law of iterated expectation. \hfill
\end{proof}

\section{Proofs in \texorpdfstring{\Cref{sec:evaluating,sec:pretrend}}{Sections 4 and 5}} \label{apx:proofs2}
\appendixnumbering{B}

\begin{proof}[Proof of Proposition \ref{thm:evaluating-obs}]
With potential outcomes and treatment rules in (\ref%
{eq:po-choice}) and (\ref{eq:treatment-choice}), we have:%
\begin{eqnarray*}
&&\E[Y_{2}(d_{1},d_{2})\mid D_{2}=d_{2},Y_{1}=y_{1},D_{1}=d_{1},Y_{0}=y_{0}] \\
&=&\E\left[ \lambda _{2}(d_{1},d_{2},y_{0},y_{1},\alpha (d_{2}),\varepsilon
_{2})\left\vert 
\begin{array}{c}
\phi _{2}(y_{0},y_{1},d_{1},\xi_1(d_{1}),\eta _{2})=d_{2}, \\ 
\lambda _{1}(d_{1},y_{0},\alpha (d_{1}),\varepsilon _{1})=y_{1}, \\ 
\phi _{1}(y_{0},\xi_0,\eta _{1})=d_{1},\text{ \ }Y_{0}=y_{0}%
\end{array}%
\right. \right]  \\
&=&\E[\lambda _{2}(d_{1},d_{2},y_{0},y_{1},\alpha (d_{2}),\varepsilon
_{2})\mid \lambda _{1}(d_{1},y_{0},\alpha (d_{1}),\varepsilon
_{1})=y_{1},Y_{0}=y_{0}]\text{,}
\end{eqnarray*}%
where the last equality is due to (\ref{assn:CI}) and the right-hand side does not vary with $d_{2}$. 

Also note: 
\begin{eqnarray*}
&&\E[Y_{2}(d_{1},d_{2})\mid D_{1}=d_{1},Y_{0}] \\
&=&\E[\lambda _{2}(d_1,d_2,y_{0},Y_{1}(d_{1}),\alpha (d_{2}),\varepsilon
_{2})\mid \phi _{1}(y_{0},\xi_0,\eta _{1})=d_{1},Y_{0}=y_{0}] \\
&=&\E[\lambda _{2}(d_1,d_2,y_{0},Y_{1}(d_{1}),\alpha (d_{2}),\varepsilon
_{2})\mid Y_{0}=y_{0}]\text{,}
\end{eqnarray*}%
where the last equality holds because of (\ref{assn:CI}) and the fact that $Y_{1}(d_{1})$ is a function of $(Y_{0},\alpha (d_{1}),\varepsilon_{1})$. Again, the right-hand side does not vary with $d_{1}$.

Furthermore, 
\begin{eqnarray*}
    \E[Y_{1}(d_{1})\mid D_{1}=d_{1},Y_{0}]  & = & \E[\lambda _{1}(d_{1},y_{0},\alpha (d_{1}),\varepsilon _{1})\mid \phi_{1}(y_{0},\xi_0,\eta _{1})=d_{1},Y_{0}=y_{0}] \\
     & = & \E[\lambda _{1}(d_{1},y_{0},\alpha (d_{1}),\varepsilon_{1})\mid Y_{0}=y_{0}]\text{,}
\end{eqnarray*}%
where the last equality is due to (\ref{assn:CI})\ and the right-hand side does not vary with $d_{1}$. Thus, mean sequential exchangeability in Assumption \ref{assn:mse}b holds.

On the other hand, under (\ref{assn:CI}),
\begin{eqnarray*}
&&\E[Y_{2}(0,0)-Y_{1}(0)\mid (D_{1},D_{2})=(d_{1},d_{2}),Y_{0}=y_{0}] \\
&=&\E\left[ \lambda _{2}(0,0,y_{0},Y_{1}(0),\alpha (0),\varepsilon
_{2})-Y_{1}(0)\left\vert 
\begin{array}{c}
\phi _{2}(y_{0},Y_{1}(d_{1}),d_{1},\xi_1(d_{1}),\eta _{2})=d_{2}, \\ 
\phi _{1}(Y_{0},\xi_0,\eta _{1})=d_{1},Y_{0}=y_{0}%
\end{array}%
\right. \right] 
\end{eqnarray*}
where $Y_{2}(0,0)-Y_{1}(0)$ is a function of $(Y_{0},\alpha (0),\varepsilon_{1},\varepsilon_{2})$. Under (\ref{assn:CI}), the right-hand side \textit{does} generally vary with $(d_{1},d_{2})$ even conditional on $Y_{0}$. 
This is because the event \textquotedblleft $(D_{1},D_{2})=(d_{1},d_{2})$\textquotedblright\ restricts the support of $(Y_{0},\alpha (d_{1}),\varepsilon_{1},\xi_1(d_{1}),\eta _1,\eta_2)$.\footnote{%
    Note $Y_{1}(d_{1})=\lambda _{1}(d_{1},Y_{0},\alpha (d_{1}),\varepsilon_{1})$ enters the conditioning event through $\phi _{2}(\cdot )$. } 
Hence, even after conditioning on $Y_{0}=y_{0}$, the right-hand side is in general non-degenerate in $(d_{1},d_{2})$. \hfill
\end{proof}

\vspace{.5cm}

\begin{proof}[Proof of Lemma \ref{lem:ex3-rule}]
\emph{(Period 2.)} Denote a unit's information set at the start of $t=2$ by $ \mathcal I_2 \equiv \{ \xi,\eta_1(\cdot),D_1,Y_1,\eta_2(\cdot) \} $. Under (I1)-(I2), $\varepsilon_2$ is independent of $\mathcal I_2$ and of $(\alpha,\varepsilon_1)$, and $\E(\alpha\mid\mathcal I_2) = m_2(D_1,\xi)$. Therefore,
\[ \mu_2(d_2,Y_1,D_1,\xi)=\E(\gamma Y_1+\alpha d_2+\varepsilon_2\mid\mathcal I_2)-c\,d_2=\gamma Y_1+[m_2(D_1,\xi)-c]\,d_2 . \]
The period-2 rational treatment rule is that $D_2=1$ if and only if $\mu_2(1,Y_1,D_1,\xi)+\eta_2(1)\ge\mu_2(0,Y_1,D_1,\xi)+\eta_2(0)$, i.e., $ m_2(D_1,\xi) + \eta_2 \ge c $, which is (\ref{eq:ex3-D2}).

\emph{(Continuation value.)} 
By definition, the continuation value at $t=2$ is
\begin{align*}
    V_2(y_1,d_1,\xi) \equiv & \;  \E\,[\, \max_{d_2}\{\mu_2(d_2,y_1,d_1,\xi)+\eta_2(d_2)\}\, ] \\
     = & \; \gamma y_1+\E\,[\eta_2(0)]+\E\,[\,\max\{0,\,m_2(d_1,\xi)-c+\eta_2\}\,],
\end{align*}
where $ m_2(1,\xi) = (1-w)m_1(\xi)+wy_1 $, $\eta_2\equiv\eta_2(1) - \eta_2(0)$, and the last term is $\Psi_{\sigma_\eta}(m_2)$ because $\eta_2\sim N(0,\sigma_\eta^2)$.

\emph{(Period 1.)} The period-1 information is $(\xi,\eta_1)$. % By (\ref{eq:D1_rule}), 
The objective function under $d_1$ is
\[ \E\,[\,Y_1(d_1) \mid \xi\,]-c\,d_1+\eta_1(d_1)+\rho\,\E\,\big[\,V_2(Y_1(d_1),d_1,\xi)\mid\xi\,\big]. \]
Under the prior, $\E\,[Y_1(d_1)\mid\xi]=m_1(\xi)\,d_1$. 
By construction, $\E(\,V_2(Y_1(d_1),d_1,\xi) \mid \xi\,) = \gamma m_1(\xi)d_1+\E\,[\eta_2(0)]+\E[\,\Psi_{\sigma_\eta}(m_2(d_1,\xi))\mid \xi \,]$. 
If $ d_1 = 0 $, then $ m_2(0,\xi) = m_1(\xi) $ and the last term is $\Psi_{\sigma_\eta}[m_1(\xi)]$. 
If $ d_1 = 1 $, then $ m_2(1,\xi) = (1-w)\,m_1(\xi) + w\,Y_1(1) = (1-w)\,m_1(\xi) + w\,( \alpha + \varepsilon_1 ) $. Given $\xi$, $\alpha\sim N(m_1(\xi),v_1)$; by (I1), $\varepsilon_1\sim N(0,\sigma_\varepsilon^2)$ and is independent of $(\xi,\alpha)$. So, $\alpha+\varepsilon_1\mid\xi\sim N\,(\,m_1(\xi),\,v_1+\sigma_\varepsilon^2\,)$. 
Since $m_2(1,\xi)$ is affine in $\alpha+\varepsilon_1$ given $\xi$, 
\begin{equation} \label{post-m2}
m_2(1,\xi) \mid \xi \sim N\big(m_1(\xi),\,wv_1\big),
\end{equation} 
because $w^2(v_1+\sigma_\varepsilon^2)=wv_1$.
By definition, $\Psi_{\sigma_\eta}(k)=\E(\max\{0,\,k-c+\sigma_\eta Z\})$ with $Z\sim N(0,1)$. 
Since $\eta_2\sim N(0,\sigma_\eta^2)$ is independent of $(\xi,\alpha,\varepsilon_1)$ by (I1), and $m_2(1,\xi)$ is a function of $(\xi,\alpha,\varepsilon_1)$, we have $\E[\max\{0,\,m_2(1,\xi)-c+\eta_2\}\mid\xi,\alpha,\varepsilon_1]=\Psi_{\sigma_\eta}(m_2(1,\xi))$ (where an argument $y_1=\alpha+\varepsilon_1$ is suppressed in $m_2(1,\xi)$ for simplicity). By iterated expectations,
\[ \E[\Psi_{\sigma_\eta}(m_2(1,\xi))\mid\xi]=\E\big[\max\{0,\,m_2(1,\xi)-c+\eta_2\}\mid\xi\big]. \]
By \eqref{post-m2} and the independence of $\eta_2\sim N(0,\sigma_\eta^2)$ from $(\xi,\alpha,\varepsilon_1)$, we have
$$ m_2(1,\xi)-c+\eta_2\mid\xi\sim N(m_1(\xi)-c,\,\bar\sigma^2),$$ 
with $\bar\sigma^2=wv_1+\sigma_\eta^2$. 
That is, given $\xi$, $m_2(1,\xi)-c+\eta_2$ is distributed as $m_1(\xi)-c+\bar\sigma Z$. Therefore
\[ \E[\Psi_{\sigma_\eta}(m_2(1,\xi))\mid\xi]=\E[\max\{0,\,m_1(\xi)-c+\bar\sigma Z\}]=\Psi_{\bar\sigma}(\,m_1(\xi)\,). \]
Subtracting the period-1 objective under $d_1=0$ from that under $d_1=1$ gives 
\begin{align*}
    & \; m_1 - c + \eta_1(1) +\rho\,[\,\gamma m_1 + \E(\eta_2(0)) + \Psi_{\bar\sigma}(m_1)\,] - \{\,\eta_1(0) + \rho\,[\,\E(\eta_2(0)) + \Psi_{\sigma_\eta}(m_1)\,]\,\} \\
    = & \; m_1+\eta_1-c+\rho\gamma m_1 + \rho\,[\,\Psi_{\bar\sigma}(m_1) - \Psi_{\sigma_\eta}(m_1)\,] =h(m_1) + \eta_1 -c.
\end{align*}
Hence $D_1=1$ if and only if $h(m_1)+\eta_1\ge c$, which is (\ref{eq:ex3-D1}). \hfill
\end{proof}

\vspace{.5cm}

\begin{lemma}\label{lem:monotone-select}
Let $X$ be a non-degenerate random variable with $\E|X|<\infty$, and let $A$ be an event with $0<P(A)<1$ such that $p(x)\equiv P(A\mid X=x)$ is strictly increasing in $x$. Then $\E(X\mid A)>\E(X\mid A^c)$.
\end{lemma}

\begin{proof}[Proof of Lemma \ref{lem:monotone-select}]
By iterated expectations, $\mathrm{Cov}(X,\ind_A)=\mathrm{Cov}(X,p(X))=\frac12\E[(X-X')(p(X)-p(X'))]$, where $X'$ is an independent copy of $X$. The integrand is non-negative and strictly positive on $\{X\neq X'\}$, an event of positive probability, so $\mathrm{Cov}(X,\ind_A)>0$. 
Also $\mathrm{Cov}(X,\ind_A)=P(A)\E(X| A)-P(A)\E(X)=P(A)P(A^c)[\E(X| A)-\E(X| A^c)]$. \hfill
\end{proof}

\vspace{.5cm}

\begin{proof}[Proof of Proposition \ref{prop:ex3-SEPT}]
With $T=2$ and $Y_0$ degenerate at zero and suppressed in notation, sequential exchangeability in \Cref{assn:mse}b consists of
\begin{align}
&\E[Y_t(d^t)\mid D_1]\ \text{is invariant in } D_1 \text{ for } t=1,2 \text{ and all } d^t, \label{eq:ex3-SE1}\\
&\E[Y_2(d_1,d_2)\mid Y_1,D_1=d_1,D_2]\ \text{is invariant in } D_2 \text{ for all } (d_1,d_2), \label{eq:ex3-SE2}
\end{align}
while the parallel trends condition in (\ref{assmp:PT}) reads: $\E[Y_2(0,0)-Y_1(0)\mid D_1,D_2]$ is invariant in $(D_1,D_2)$.

Let $\lambda\equiv\sigma_{\xi\alpha}/\sigma_\xi^2$, so that $m_1(\xi)=a_\alpha+\lambda(\xi-a_\xi)$, and write the posterior at $t=2$ compactly as $ m_2(D_1,\xi) = m_1(\xi) + D_1\,w\,[\alpha+\varepsilon_1-m_1(\xi)] $.

\emph{(Evaluation of parallel trends.)} By the period-1 rational treatment rule (\ref{eq:ex3-D1}), $D_1$ is a function of $(\xi,\eta_1)$; by the period-2 rational treatment rule (\ref{eq:ex3-D2}), $D_2$ is a function of $(\xi,\eta_1,\alpha,\varepsilon_1,\eta_2)$. Neither involves $\varepsilon_2$, which is independent of $(\xi,\alpha,\varepsilon_1,\eta_1,\eta_2)$ under (I1). Therefore, $\E(\varepsilon_2\mid D_1,D_2)=0$ and by (\ref{eq:ex3-trend}),
\begin{equation}\label{eq:ex3-pf-trend}
\E[Y_2(0,0)-Y_1(0)\mid D_1,D_2]=(\gamma-1)\,\E(\varepsilon_1\mid D_1,D_2).
\end{equation}
If $\gamma=1$, the right-hand side is zero for all $(D_1,D_2)$, so (\ref{assmp:PT}) holds.

Next, consider the case with $\gamma\neq1$. On $\{D_1=0\}$, $ m_2(D_1,\xi) = m_2(0,\xi) = m_1(\xi) $ and $D_2=\ind\{m_1(\xi)+\eta_2\ge c\}$ is a function of $(\xi,\eta_2)$. 
Then on $\{D_1=0\}$, $(D_1,D_2)$ is a function of $(\xi,\eta_1,\eta_2)$, which is independent of $\varepsilon_1$ with $\E(\varepsilon_1\mid D_1=0,D_2)=0$. 
On $\{D_1=1\}$, $D_2=\ind\{U+w\varepsilon_1\ge c\}$ with $U\equiv(1-w)m_1(\xi)+w\alpha+\eta_2$. By (I1), $\varepsilon_1$ is independent of $(U,D_1)$, so $P(D_2=1\mid\varepsilon_1=e,D_1=1)=P(U\ge c-we\mid D_1=1)$. 
Because $\eta_2$ is normal and independent of $(\xi,\alpha,\eta_1)$, the conditional distribution of $U$ given $D_1=1$ has a density that is positive on $\mathbb R$, so this probability is strictly increasing in $e$ (as $w>0$) and lies strictly between $0$ and $1$. 
Applying \Cref{lem:monotone-select} with $X\equiv\varepsilon_1$ and $A\equiv\{D_2=1\}$, under the conditional law given $D_1=1$, 
gives $\E(\varepsilon_1\mid D_1=1,D_2=1)>\E(\varepsilon_1\mid D_1=1,D_2=0)$. 
It then follows from \eqref{eq:ex3-pf-trend} that $\Delta^{PT}$ is nonzero with a sign identical to that of $\gamma-1$.

\emph{(Evaluation of sequential exchangeability.)} First, consider the case with $\sigma_{\xi\alpha}=0$. Then $m_1(\xi)$ is degenerate at $a_\alpha$, so $D_1=\ind\{h(a_\alpha)+\eta_1\ge c\}$ is a function of $\eta_1$ alone and is independent of $(\alpha,\varepsilon_1,\varepsilon_2)$; hence $\E[Y_t(d^t)\mid D_1]=\E[Y_t(d^t)]$ for $t=1,2$. This establishes (\ref{eq:ex3-SE1}). 
To check (\ref{eq:ex3-SE2}), fix $ d_1 $ and condition on $D_1=d_1$ so that $Y_1=Y_1(d_1)=\alpha d_1+\varepsilon_1$ and $m_2(d_1,\xi)=a_\alpha+d_1\,w\,(Y_1-a_\alpha)$. 
Then $D_2=\ind\{a_\alpha+d_1w(Y_1-a_\alpha)+\eta_2\ge c\}$ is a function of $(Y_1,\eta_2)$. Because $(\eta_1,\eta_2)$ is independent of $(\alpha,\varepsilon_1,\varepsilon_2)$, the events $\{D_1=d_1\}$ (in $\eta_1$) and $\{D_2=d_2\}$ (in $\eta_2$ given $Y_1$) carry no information about $(\alpha,\varepsilon_2)$ beyond $Y_1$. 
Formally, since $\sigma_{\xi\alpha}=0$, $\alpha\sim N(a_\alpha,\sigma_\alpha^2)$ is independent of $\xi$, and $v_1=\sigma_\alpha^2$, so $w=\sigma_\alpha^2/(\sigma_\alpha^2+\sigma_\varepsilon^2)$. 
Hence
\begin{align*}
    \E[Y_2(d_1,d_2) \mid Y_1=y_1,D_1=d_1,D_2] = & \; \gamma y_1 + d_2\,\E(\alpha\mid\alpha d_1+\varepsilon_1=y_1)+ \E(\varepsilon_2) \\
    = & \;\gamma y_1+d_2\big[a_\alpha+d_1\,w\,(y_1-a_\alpha)\big] ,
\end{align*}
where the second equality uses $\E(\varepsilon_2)=0$ and the normal conditional-mean formula: $\E(\alpha\mid\varepsilon_1=y_1)=a_\alpha$ if $d_1=0$, and $\E(\alpha\mid\alpha+\varepsilon_1=y_1)=a_\alpha+w(y_1-a_\alpha)$ if $d_1=1$.
The right-hand side in the display above does not vary with $D_2$; this establishes \eqref{eq:ex3-SE2}.

Next, consider the case with $\sigma_{\xi\alpha}\neq0$. We show that (\ref{eq:ex3-SE2}) fails for $(d_1,d_2)=(0,1)$. 
On $\{D_1=0\}$, $Y_1=\varepsilon_1$ and $m_2(D_1,\xi)=m_2(0,\xi)=m_1(\xi)$, so $D_2=\ind\{m_1(\xi)+\eta_2\ge c\}$ and $(D_1,D_2)$ is a function of $(\xi,\eta_1,\eta_2)$. 
By (I1), $(\varepsilon_1,\varepsilon_2)$ is independent of $(\xi,\alpha,\eta_1,\eta_2)$, and $ \E(\alpha\mid\xi,\eta_1,\eta_2) = m_1(\xi) $, so
\[ \E(Y_2(0,1) | Y_1,D_1=0,D_2) = \gamma Y_1 + \E(\alpha | D_1=0,D_2) = \gamma Y_1 + \E[m_1(\xi) | D_1=0,D_2], \]
where the second equality uses the law of iterated expectation by further conditioning on $\xi$ over $\{D_1=0,D_2\}$.
Given $D_1=0$, $m_1(\xi)$ is non-degenerate in $\xi$ when $ \sigma_{\xi\alpha}\neq 0 $.
Moreover, $\Pr\{D_1=0 \mid m_1(\xi) = m_1\} = \Phi\big(\tfrac{c-h(m_1)}{\sigma_\eta}\big)\in(0,1)$; and since $\eta_2\perp(\xi,\eta_1)$, $\Pr\{D_2 = 1 \mid m_1(\xi) = m_1, D_1=0\} = \Phi\big((m_1-c)/\sigma_\eta\big)$ is strictly increasing in $m_1$ and lies in $(0,1)$. 
Applying Lemma \ref{lem:monotone-select} with $ X \equiv m_1(\xi)$ and $ A \equiv \{D_2=1\} $, under the conditional law given $D_1=0$, gives $\E[m_1(\xi)\mid D_1=0,D_2=1]>\E[m_1(\xi)\mid D_1=0,D_2=0]$. Hence $\E[Y_2(0,1)\mid Y_1,D_1=0,D_2]$ depends on $D_2$, and (\ref{eq:ex3-SE2}) fails. \hfill
\end{proof}

\bigskip

\begin{proof}[Proof of Proposition \ref{prop:pretrend}]
Let $Z^{*}\equiv\pi_{0}+ \widetilde Z-\eta$, so that $D=\ind\{Z^{*}>0\}$ by (\ref{eq:pretrend-rule}) and $f^{*}$ is the density of $Z^{*}$.
Let $S\in\mathbb R$ be a generic linear combination of $(Y_{0},\xi,\alpha,\varepsilon_{1})$.

Under \Cref{assn:pretrend}, $(Y_{0},\xi,\alpha,\varepsilon_{1},\eta)$ is multivariate normal, so $Z^{*}$ and $S$ are bivariate normal, with $\mathrm{Var}(Z^{*})\geq\mathrm{Var}(\eta)=1$. (Hence the density of $Z^*$ at $0$ is positive, i.e., $f^{*}(0)>0$.)
By bivariate normality, $\E(S\,|\,Z^{*})=\E(S)+\tfrac{\mathrm{Cov}(S,Z^{*})}{\mathrm{Var}(Z^{*})}(Z^{*}-\mu^*)$, where $ \mu^* \equiv \E(Z^{*}) $. 
Moreover, by the normality of $Z^*$, $\E[\,(Z^{*}-\mu^*)\ind\{Z^{*}>0\}\,] = \mathrm{Var}(Z^{*})f^{*}(0)$.\footnote{
    To see this, let $\sigma_{*}^{2}\equiv\mathrm{Var}(Z^{*})$. The normal density of $Z^*$ 
    satisfies $(z-\mu^{*})f^{*}(z)=-\sigma_{*}^{2}\,f^{*\prime}(z)$, so
    \begin{align*}
        \E[\,(Z^{*}-\mu^{*})\ind\{Z^{*}>0\}\,]
        &=\int_{0}^{\infty}(z-\mu^{*})f^{*}(z)\,dz
        =-\sigma_{*}^{2}\int_{0}^{\infty}f^{*\prime}(z)\,dz % \\
        =\sigma_{*}^{2}f^{*}(0).
    \end{align*} }

With $ D $ determined by $Z^*$, the law of iterated expectations implies $\mathrm{Cov}(D,S)=\mathrm{Cov}\big(D, \E(S\mid Z^{*}) \big)$. Substituting the linear form of $\E(S\mid Z^{*})$, we get
\begin{align}
    \mathrm{Cov}(D,S)
    &=\tfrac{\mathrm{Cov}(Z^{*},S)}{\mathrm{Var}(Z^{*})}\,\mathrm{Cov}\big(\ind\{Z^{*}>0\}, Z^{*}-\mu^{*} \big)
    =\tfrac{\mathrm{Cov}(Z^{*},S)}{\mathrm{Var}(Z^{*})}\,\E\{\,\ind\{Z^{*}>0\}[Z^{*}-\mu^{*}]\,\} \notag\\
    &=\mathrm{Cov}(Z^{*},S)\,f^{*}(0)=\mathrm{Cov}(\widetilde Z,S)\,f^{*}(0), \label{eq:pretrend-stein}
\end{align}
where the second equality holds because $\E(Z^{*}-\mu^{*})=0$, the third by the identity above, and the last because $\eta$ is independent of $S$ and $\pi_{0}$ is a constant.

Applying (\ref{eq:pretrend-stein}) with $S=Y_{1}-Y_{0}$ gives $$\mathrm{Cov}(D,Y_{1}-Y_{0})= \mathrm{Cov}(\widetilde Z,Y_1-Y_0)\,f^{*}(0),$$
where $ \mathrm{Cov}(\widetilde Z,Y_{1}-Y_{0}) \equiv A(\gamma,\pi,\bsig) $ is linear in $\bsig$ given $\gamma,\pi$. Specifically, 
\begin{align*} % \label{eq:pretrend-A}
    A(\gamma,\pi,\bsig) = & \; \pi_{1}\mathrm{Cov}(Y_{0},Y_{1}-Y_{0})+\pi_{2}\mathrm{Cov}(Y_{1},Y_{1}-Y_{0})+\pi_{\xi}'\mathrm{Cov}(\xi,Y_{1}-Y_{0}) \notag\\
    = &\; \pi_{1}(\sigma_{0\alpha}-\tilde\gamma\sigma_{0}^{2})
    +\pi_{2}[\sigma_{\alpha}^{2}+\sigma_{\varepsilon}^{2}-(\tilde\gamma - \gamma)\sigma_{0\alpha}-\gamma\tilde\gamma\sigma_{0}^{2}] 
      + \pi_{\xi}'(\sigma_{\xi\alpha}-\tilde\gamma\sigma_{0\xi}), 
\end{align*}
where $\tilde\gamma\equiv 1-\gamma$.

By construction, $Y_{2}(0)-Y_{1}=\gamma(Y_{1}-Y_{0})+\varepsilon_{2}-\varepsilon_{1}$.
Here $\mathrm{Cov}(D,\varepsilon_{2})=0$, because $D$ is a function of $(Y_{0},\xi,\alpha,\varepsilon_{1},\eta)$, which is independent of $\varepsilon_{2}$ under \Cref{assn:pretrend}(a). 
Besides, $\mathrm{Cov}(D,\varepsilon_{1})=\mathrm{Cov}(\widetilde Z,\varepsilon_{1})\,f^{*}(0)=\pi_{2}\,\sigma_{\varepsilon}^{2}\,f^{*}(0)$ by (\ref{eq:pretrend-stein}) with $S=\varepsilon_{1}$.
Hence $$\mathrm{Cov}\big(D,Y_{2}(0)-Y_{1}\big)=\big[\gamma\,A(\gamma,\pi,\bsig)-\pi_{2}\,\sigma_{\varepsilon}^{2}\big]f^{*}(0).$$

By (\ref{eq:pretrend-delta}), $\hpr$ holds if and only if $\mathrm{Cov}(D,Y_{1}-Y_{0})=0$ and $\hpt$ holds if and only if $\mathrm{Cov}(D,Y_{2}(0)-Y_{1})=0$.
With $f^{*}(0)>0$, $\hpr$ holds if and only if $A(\gamma,\pi,\bsig)=0$, and $\hpt$ holds if and only if $\gamma\,A(\gamma,\pi,\bsig)=\pi_{2}\,\sigma_{\varepsilon}^{2}$.
\end{proof}

\section{Supplement to \texorpdfstring{\Cref{sec:causal-no-se}}{Section 3.2}} \label{apx:supp-causal-no-se}
\appendixnumbering{C}

\subsection{First-order autoregressive potential outcomes} \label{apx:ARPO}

This subsection considers a dynamic panel model for {\it all} potential outcomes (including treated as well as untreated ones), where treatment effects are homogeneous, treatments are sequentially exogenous, and the shocks in potential outcomes are (conditionally) serially uncorrelated. In this case, the first-differenced IV regression in (\ref{eq:fd-iv}) identifies the homogeneous treatment effect, \emph{even when} sequential exchangeability does not hold.%
\footnote{
With a minor abuse of terminology, in this benchmark we use treatment effect homogeneity to refer to the common conditional mean $\beta$ under Assumption~\ref{assn:arpo}; note that individual effects may still vary through the treatment-indexed shocks, i.e., $\varepsilon_{it}(d^{t-1},1)-\varepsilon_{it}(d^{t-1},0)$.
}

Suppose the \emph{potential} outcomes follow a first-order autoregressive process with a homogeneous causal effect as follows. 

\vspace{.25cm}

\begin{assumption} \label{assn:arpo}
    For all $d^t\in\{0,1\}^t$, 
    \begin{equation} \label{eq:arpo}
        Y_{it} (d^t) = \beta d_{t} + \gamma Y_{it-1} (d^{t-1}) + \theta_t + \alpha_i + \varepsilon_{it} (d^t) \text{  for  } t \geq 1,
    \end{equation}
    with  $Y_{i0} (d^0) \equiv Y_{i0}$,\footnote{Let $d^0$ denote the null vector so that $Y_{i0}(d^0)\equiv Y_{i0}$ conforms with the notation $Y_{it}(d^t)\equiv Y_{it}(d_1,...,d_t)$.} where $\beta, \gamma$, and $\theta_t$ are constants, $\alpha_i$ are individual fixed effects, and $\varepsilon_{it}(d^t)$ are time-varying shocks in potential outcomes s.t. 
    \begin{equation} \label{eq:arpo-exog}
        \E[ \varepsilon_{it} (\cdot) \mid  D_i^t,Y_{i0}, \varepsilon^{t-1}_i(\cdot), \alpha_i ] = 0 \text{ for all } t\geq 1.
    \end{equation}
\end{assumption}

As in the text, let $\varepsilon_i^{t} (d^{t}) \equiv (\varepsilon_{i1} (d^1), \dots, \varepsilon_{it} (d^{t}))$ for $t \geq 1$ and $\varepsilon_i^{t} (\cdot)$ the collection of vectors of potential shocks with $\varepsilon^{0}_i(\cdot)$ being a vacuous null vector.
The condition \eqref{eq:arpo-exog} combines sequential exogeneity and serial uncorrelation, conditional on fixed effects. Assumption \ref{assn:arpo} is consistent with the condition on potential outcomes in Assumption \ref{assn:po-treatments}.

\vspace{.25cm}

\begin{observation} \label{thm:id-linear-po}
    Under \Cref{assn:arpo}, the observed outcomes $Y_{it} = Y_{it} (D_{i}^t)$ satisfy \eqref{eq:dyn-panel}  with reduced-form errors $\varepsilon_{it} \equiv \varepsilon_{it} (D_{i}^t)$.
\end{observation}

\vspace{.25cm}

\begin{lemma}
\label{thm:Y-eps}
Suppose \Cref{eq:arpo} holds, and let $\nu_{it} (d^t) \equiv \alpha_i + \varepsilon_{it} (d^t)$. 
For any $d^t \in \{0,1\}^t$ and $ t \geq s \geq 1 $,
\begin{align}
\label{eq:eps-Y}
    \nu_{it} (d^t) 
    &=
    Y_{it} (d^t) - \gamma Y_{it-1} (d^{t-1}) - \beta d_t - \theta_t,  \\
\label{eq:Y-eps}    
    Y_{it} (d^t) & = \gamma^{t-s+1} Y_{is-1} (d^{s-1}) + \sum\nolimits_{r=s}^t
    \gamma^{t-r} \, [ \beta d_{r} + \theta_r + \nu_{ir} (d^r) ].
\end{align}
\end{lemma}

\begin{proof}[Proof of Lemma \ref{thm:Y-eps}]
The first equality \eqref{eq:eps-Y} is immediate. The second equality \eqref{eq:Y-eps} follows by induction on $t$ given $s$. 
For any $s \geq 1$, the base case $t = s$ repeats \eqref{eq:arpo} and is thus immediate.
Now suppose \eqref{eq:Y-eps} holds for some index $t-1 \geq s$:
\[
    Y_{it-1} (d^{t-1}) 
    =
    \gamma^{t-s} Y_{is-1} (d^{s-1}) + 
    \sum\nolimits_{r=s}^{t-1}
    \gamma^{t-1-r} \times
    [
    \beta d_{r} + \theta_r + \nu_{ir} (d^r)
    ]
\]
It suffices to show the formula also holds for $t$. For $t$, the potential outcome model states:
\[
Y_{it} (d^t) = 
\gamma Y_{it-1} (d^{t-1}) + \beta d_{t} + \theta_t + \nu_{it} (d^t)
\]
Plugging in from the preceding equation: 
\begin{align*}
Y_{it} (d^t) = &
\gamma \left( 
    \gamma^{t-s} Y_{is-1} (d^{s-1}) + 
    \sum\nolimits_{r=s}^{t-1}
    \gamma^{t-1-r} \times
    [
    \beta d_{r} + \theta_r + \nu_{ir} (d^r)
    ]  \right)  \\
& + \beta d_{t} +\theta_t + \nu_{it} (d^t)
\end{align*}
Collecting terms yields the desired result. 
\end{proof}

\vspace{.25cm}

Next, we show \Cref{assn:arpo} implies a ``conditional'' version of mean sequential exchangeability after controlling for fixed effects.

\vspace{.25cm}

\begin{proposition}
\label{thm:seq-exog-implies-seq-exchange-fe}
    Under Assumption \ref{assn:arpo}, mean sequential exchangeability holds conditional on fixed effects:
    \begin{equation}
    \label{eq:mse-fe}
    \E[ Y_{it} (d^t) \mid  D_{is} = d_s^{\,\prime}, Y_i^{s-1}, D_i^{s-1} = d^{s-1}, \alpha_i]
    \quad \text{is invariant in $d_s^{\,\prime}$ for $t \geq s \geq 1$}
    \end{equation}
    for all $d^t \in \{0,1\}^t$, where $d^{s-1}$ consists of the first $s-1$ components in $d^t$.
\end{proposition}

\begin{proof}[Proof of Proposition \ref{thm:seq-exog-implies-seq-exchange-fe}]
For $t \geq s \geq 1$, we have:
\begin{align*}
    &
    \E[ Y_{it} (d^t) \mid  \alpha_i, Y_i^{s-1}, D_i^{s-1} = d^{s-1}, D_{is} = d_s']
    \\
    =& 
    \E[ Y_{it} (d^t) \mid  \alpha_i, Y_i^{s-1} (d^{s-1}), D_i^{s-1} = d^{s-1}, D_{is} = d_s']
    \\[0.1in]
    =&
    \gamma^{t-s+1} Y_{is-1} (d^{s-1}) + \\
    & \sum\nolimits_{r=s}^t \gamma^{t-r}\,
    \{   \beta d_{r} + \theta_r + \alpha_i + \E[\varepsilon_{ir} (d^r)
         \mid  \alpha_i, Y_i^{s-1} (d^{s-1}), D_i^{s-1} = d^{s-1}, D_{is} = d_s'] \}
\end{align*}
where the first equality follows by substituting realized treatment, and the second by \eqref{eq:Y-eps} of Lemma \ref{thm:Y-eps}.
It therefore suffices to show that: 
\[
    \E[\varepsilon_{it} (d^t)
    \mid  \alpha_i, Y_i^{s-1} (d^{s-1}), D_i^{s-1} = d^{s-1}, D_{is} = d_s']
\]
is invariant in $d_s'$, or a fortiori that the expectation equals zero, for $t \geq s$. 
Lemma \ref{thm:Y-eps} implies that the preceding expression is equal to:
\[
    \E[\varepsilon_{it} (d^t)
    \mid  \alpha_i, Y_{i0}, D_i^{s-1} = d^{s-1}, D_{is} = d_s', \nu_i^{s-1} (d^{s-1})]
\]
because $\nu_i^{s-1} (d^{s-1})$ can be recovered from $Y_i^{s-1} (d^{s-1})$ by \eqref{eq:eps-Y}, and conversely $Y_i^{s-1} (d^{s-1})$ can be recovered from $(Y_{i0}, \nu_i^{s-1} (d^{s-1}))$ by \eqref{eq:Y-eps}.
Since $\varepsilon_i^{s-1} (d^{s-1}) = \nu_i^{s-1} (d^{s-1}) - \alpha_i $, this is equivalent to:
\[
    \E[\varepsilon_{it} (d^t)
    \mid  \alpha_i, Y_{i0}, D_i^{s-1} = d^{s-1}, D_{is} = d_s', \varepsilon_i^{s-1} (d^{s-1})]
\]
Iterating expectations and invoking \eqref{eq:arpo-exog} prove the claim.
\end{proof}

\vspace{.25cm}

Even under Assumption \ref{assn:arpo}, the conditional mean in (\ref{eq:mse-fe}) is in general non-degenerate in the fixed effects $\alpha_i$, whose distribution conditional on $Y_i^{s-1},D_i^{s}$ may depend on $D_{is}$. 
In such cases, Assumption \ref{assn:mse}(b) does not hold in general, because \emph{both} potential outcomes and contemporary treatments depend on the fixed effects non-trivially.

Despite possible violation of sequential exchangeability, \Cref{assn:arpo} is sufficient for identifying $(\beta, \gamma, \Delta \theta_2)$ from the first-differenced IV estimand in \eqref{eq:fd-iv}. 

\vspace{.25cm}

\begin{proposition} \label{thm:suff-cond-for-AR1991}
    Under Assumption \ref{assn:arpo}, the IV estimand $(\widetilde \beta, \widetilde \gamma, \widetilde{\Delta\theta_2})$ in \eqref{eq:fd-iv} equals $(\beta, \gamma, \Delta \theta_2)$ in \eqref{eq:arpo} when $\E(Z'X)$ is non-singular. 
\end{proposition}

\vspace{.25cm}

\begin{proof}[Proof of Proposition \ref{thm:suff-cond-for-AR1991}]
Per Observation \ref{thm:id-linear-po}, when the potential outcomes satisfy \Cref{assn:arpo}, the implied observables $(D_i^t,Y_i^t)$ satisfy the reduced-form model of observed outcomes in (\ref{eq:dyn-panel}) with coefficients $ (\gamma, \beta, \theta_t)$, fixed effects $\alpha_i$ and reduced-form errors $\varepsilon_{it} = \varepsilon_{it} (D_{i}^t)$. 
First-differencing yields
\[\Delta Y_{it} = \beta\Delta D_{it} +\gamma\Delta Y_{it-1} + \Delta \theta_t + \Delta \varepsilon_{it},\]
where \( \E(\varepsilon_{it}\vert D_i^t,Y_i^{t-1}) = 0 \) because of \eqref{eq:arpo-exog} in \Cref{assn:arpo} and the fact that $Y_i^{t-1}$ is a function of $\{D_i^{t-1},Y_{i0},\varepsilon_i^{t-1}(\cdot),\alpha_i\}$.
Consequently, \(\E(\Delta \varepsilon_{it} \vert D_i^{t-1},Y_i^{t-2}) = 0\) for all $t\geq 2$.
Using $\E(\Delta D_{it} \vert D_i^{t-1},Y_i^{t-2})$ and $\E(\Delta Y_{it-1} \vert D_i^{t-1},Y_i^{t-2})$ as instruments for $\Delta D_{it}$ and $\Delta Y_{it-1} $ in first-differenced IV regression proves the result. 
\end{proof}

\bigskip

The intuition for Proposition \ref{thm:suff-cond-for-AR1991} is that the conditions on the \emph{potential} outcome process in \Cref{assn:arpo} ensure that the implied \emph{observed} outcome process satisfies the conditions in \cite{arellanobond91} so that the first-differenced IV estimand identifies the constant causal parameters. 
Generalization to later trends $\Delta \theta_t$ for $t\geq 3$ only requires using $X_{it}\equiv(\Delta D_{it},\Delta Y_{it-1}, 1)$ and $Z_{it}\equiv \E(X_{it}\vert Y_i^{t-2},D_i^{t-1})$ in first-differenced IV regression.

\subsection{Extension: ATT under general treatment history} \label{apx:gen-arupo}

This subsection extends the identification of average treatment effects on the treated to a \emph{general} sequence of past treatments $d^{t-1}$, beyond the case of no prior treatment.

For any $d^{t-1}\in\{0,1\}^{t-1}$, generalize the definition of $\mathrm{ATFT}$ in \eqref{eq:atft} to
\[ \mathrm{ATT}_t(d^{t-1},y_0)\equiv \E[\tau_{it}(d^{t-1})\vert D_{it}=1, \mathcal{S}_{t-1} (d^{t-1}, y_0) ] \]
where, with slight abuse of notation, we let $ \mathcal{S}_{t-1} (d^{t-1}, y_0) \equiv \{ D_i^{t-1}=d^{t-1}, Y_{i0}=y_0 \}$. We maintain \Cref{assn:arupo}, but strengthen \Cref{assn:MI-hte} as follows.

\begin{assumption} \label{assn:MI-hte2}
    For any initial condition $y_0$, $t\geq 1$ and $d^{t-1}\in\{0,1\}^{t-1}$,
    \[  \E[\,\tau_{it}(d^{t-1})\vert D_{it}=1, \mathcal S_{t-1}(d^{t-1},y_0), \tau_i^{t-1}(d^{t-2}),\alpha_i,\varepsilon_i^{t-1}(\cdot)\,]
        =  \E[\,\tau_{it}(d^{t-1})\vert D_{it}=1, \mathcal S_{t-1}(d^{t-1},y_0)\,],  \]
    where $d^{t-1}\equiv (d^{t-2},d_{t-1})$.
\end{assumption}

\noindent
Assumption \ref{assn:MI-hte2} resembles Assumption \ref{assn:MI-hte} except that it also posits conditional mean independence of the current treatment effect from the past treatment effects $\tau_i^{t-1}(d^{t-2})$ along the realized history. 
Under Assumption \ref{assn:MI-hte2}, 
\begin{equation}\label{eq:MI-tau2}
    \E[\tau_{it}(d^{t-1})\vert D_{it}=1, D_i^{t-1} = d^{t-1},Y_i^{t-1}] =
        \E[\tau_{it}(d^{t-1})\vert D_{it}=1, D_i^{t-1} = d^{t-1}, Y_{i0}],
\end{equation}
because $Y_i^{t-1}$ is a function of \( \{Y_{i0},\tau_i^{t-1}(d^{t-2}),\varepsilon_i^{t-1}(\cdot),\alpha_i\} \) conditional on $D_i^{t-1}=d^{t-1}$. Define new instruments:
\[ Z_{it} (d^{t-1},y_0) \equiv \E[\,X_{it}\,\vert\, H^*_{it-1}, \mathcal S_{t-1} (d^{t-1},y_0)\,]. \]

As in \Cref{sec:causal-no-se}, we suppress the realized initial condition $y_0$ as an argument in the notation of the instrument $Z_{it}(d^{t-1},y_0)$, the conditioning event $ \mathcal S_{t-1}(d^{t-1},y_0) $, and the target parameter $\mathrm{ATT}_t(d^{t-1},y_0)$. Consequently, the rank condition and the result in \Cref{thm:arupo-g-iv} are stated as conditioning on a generic realization of $y_0$.

Define the (conditional) IV estimand under general treatment history as
\begin{equation} \label{eq:IV-estimand-gen}
    \begin{pmatrix}
        \widetilde \beta (d^{t-1}) \\ \widetilde \gamma (d^{t-1}) \\ \widetilde{\Delta \theta_t} (d^{t-1})
    \end{pmatrix}
    \equiv 
    \E[\,Z_{it}' (d^{t-1}) X_{it}\,\vert\, \mathcal S_{t-1} (d^{t-1})\,]^{-1}\,\E[\,Z_{it}' (d^{t-1}) \Delta Y_{it}\,\vert\, \mathcal S_{t-1} (d^{t-1})\,]. 
\end{equation}

\begin{proposition} \label{thm:arupo-g-iv}
Suppose \Cref{assn:arupo} and \Cref{assn:MI-hte2} hold, and:
\[  \E[\,Z_{it}(d^{t-1})^{\,\prime}Z_{it}(d^{t-1})\,\vert\, \mathcal S_{t-1}(d^{t-1})\,] \] 
has full rank for $t\geq 3$. Then the slope estimands defined in (\ref{eq:IV-estimand-gen}) recover the following causal parameters for $ t \geq 3 $:
\begin{equation*}
    \widetilde \beta (d^{t-1}) = \mathrm{ATT}_t(d^{t-1}) ; \qquad \widetilde \gamma (d^{t-1}) = \gamma,
\end{equation*}
and the intercept estimand recovers: 
\[
\widetilde{\Delta \theta_t} (d^{t-1}) =
\begin{cases}
    \Delta \theta_t & \text{when $d_{t-1} = 0$}; \\
    \Delta\theta_t + \mathrm{ATT}_t (d^{t-1}) - \mathrm{ATT}_{t-1} (d^{t-2}) & \text{when $d_{t-1} = 1$.}
\end{cases}
\]
\end{proposition}

\vspace{.25cm} 

\begin{proof}[Proof of Proposition \ref{thm:arupo-g-iv}]
This proof resembles that of \Cref{thm:IV-arupo}; we only note the differences. 
First, consider the case with $d_{t-1} = 0$ so that $\Delta D_{it}= D_{it}$ conditional on $\mathcal S_{t-1}(d^{t-1})$. 
The proof in this case is similar to that of \Cref{thm:IV-arupo}, except for two changes. 
One is that it conditions on $\mathcal S_{t-1} (d^{t-1})$ instead of $\mathcal S_{t-1} \equiv \mathcal S_{t-1}(0^{t-1})$ in \Cref{thm:IV-arupo}.
The other is that $H_{it-1}$ (and $H^*_{it-1}$) is now a function of $\{\varepsilon_i^{t-2} (\cdot), \alpha_i, \tau_i^{t-2} (d^{t-3})\}$ on $\mathcal S_{t-1} (d^{t-1})$, which requires \Cref{assn:MI-hte2}, or its implication \eqref{eq:MI-tau2}, since $H^*_{it-1}$ is a function of $(D_i^{t-1},Y_i^{t-2})$, so that: 
\[ \E[\, \tau_{it}(D_i^{t-1}) \mid  D_{it} = 1, H^*_{it-1}, \mathcal S_{t-1} (d^{t-1})\,]
     =  \mathrm{ATT}_t (d^{t-1}). 
\]

In the other case, $d_{t-1} = 1$ so that $\Delta D_{it}= D_{it} - 1$ conditional on $\mathcal S_{t-1}(d^{t-1})$.
Relative to Proposition \ref{thm:IV-arupo}, the conditional mean of the differenced outcome (analogous to \eqref{eq:arupo-cond-mean}) now includes an additional term: 
\[
- \E[\, \tau_{it-1} (D_i^{t-2}) \mid  H^*_{it-1}, \mathcal S_{t-1} (d^{t-1})\,] = - \mathrm{ATT}_{t-1} (d^{t-2});
\]
the equality holds under \Cref{assn:MI-hte2}, applied at $t-1$ with $D_{it-1}=d_{t-1}=1$ on $\mathcal S_{t-1}(d^{t-1})$, because $H_{it-1}$ (and $H^*_{it-1}$) is a function of $\varepsilon_i^{t-2} (\cdot)$, $\alpha_i$ and $\tau_i^{t-2} (d^{t-3})$ conditional on $\mathcal S_{t-1} (d^{t-1})$. Plugging in, with $\E[D_{it}\mid H^*_{it-1},\mathcal S_{t-1}(d^{t-1})]=\E[\Delta D_{it}\mid H^*_{it-1},\mathcal S_{t-1}(d^{t-1})]+1$ in the first term, and rearranging gives:
\begin{align*}
\E[\Delta Y_{it}\vert H^*_{it-1},\mathcal S_{t-1}(d^{t-1})] 
=& \mathrm{ATT}_t(d^{t-1}) \E[ \Delta D_{it} \mid  H^*_{it-1}, \mathcal S_{t-1} (d^{t-1})] + \\
& \gamma\E[\Delta Y_{it-1}\vert H^*_{it-1},\mathcal S_{t-1}(d^{t-1})] + \Delta\theta_t + \\
& \mathrm{ATT}_t (d^{t-1}) - \mathrm{ATT}_{t-1} (d^{t-2}).
\end{align*} 
The rest of the proof is similar to that in the first case above, except that the estimand for the intercept now includes an intertemporal difference in treatment effects: 
\( \widetilde{\Delta \theta_t} = \Delta\theta_t + \mathrm{ATT}_t (d^{t-1}) - \mathrm{ATT}_{t-1} (d^{t-2}) \). \hfill
\end{proof}

\section{Supplement to \texorpdfstring{\Cref{sec:evaluating-learning}}{Section 4.2}} \label{apx:supp-learning}
\appendixnumbering{D}

\subsection{Rational Treatment with Learning: A Structural Model} \label{apx:structural}

This subsection shows how the treatment rule in \eqref{eq:treatment-choice} arises in a structural model where forward-looking units choose optimal treatment sequences, using past outcomes to infer the unknown fixed effects in potential outcomes.

Consider a model with $T=2$. Each individual's \textit{per-period, ex post} payoff from choosing $D_1 = d_1$ at $t=1$ is: 
\[ u_1(d_1,Y_1(d_1)) + \eta_1(d_1).\]
Moreover, given $D_1$, the instantaneous, ex post payoff from choosing $D_2 = d_2$ at $t=2$ is:
\[u_2(d_2,Y_2(D_1,d_2)) + \eta_2(d_2).\]
Potential outcomes $Y_t(\cdot)$ are determined by $(Y_0,\varepsilon_1,\varepsilon_2)$ and individual fixed effects $ \alpha \equiv (\alpha(0),\alpha(1)) $ via deterministic functions $\lambda_t(\cdot)$ as defined in (\ref{eq:po-choice}).
Individuals know all functions $u_{t}(\cdot)$ and $\lambda_t(\cdot)$, but never observe the potential fixed effects $\alpha$ or idiosyncratic errors $\varepsilon_t$ in $Y_t(\cdot)$.
In each period $t$, individuals observe $\eta_t\equiv (\eta_t(0),\eta_t(1))$, and know the distribution of the idiosyncratic errors $ \varepsilon_t $. 

At the start of the first period $ t = 1 $, each individual  holds a prior belief for $ \alpha $. 
This prior is characterized by a parameter vector $ \varphi_1(Y_0,\xi) $, where $ \xi $ is a persistent fixed effect known privately to the individual unit. 
Note here we specify the unobserved heterogeneity in the prior and posterior simply as $ \xi_0 = \xi_1(\cdot) \equiv \xi $. 

For this illustration, suppose that $(\alpha,Y_0,\xi)$, $\varepsilon_1$, $\varepsilon_2$, $\eta_1(\cdot)$, and $\eta_2(\cdot)$ are mutually independent, that these independence restrictions are known to individuals, and that expected payoffs are finite.

At the start of $ t = 2 $, the individual uses the observed $Y^1$ and the realized treatment $D_1$ to update the posterior for $\alpha$, which is characterized by $\varphi_2(Y^1,D_1,\xi)$.

In period $t=2$, a rational individual with learning chooses:
\begin{equation}\label{eq:D2_rule-apx}
    D_2=\argmax_{d_2\in\{0,1\}} \underset{\mu_{2}(d_2,Y^1,D_1,\xi)}{\underbrace{\int u_2(d_2,y_2)dF_{Y_2(D_1,d_2)\vert Y^1,D_1,\xi}(y_2)}} + \eta_2(d_2),
\end{equation}
where $F_{Y_2(D_1,d_2)\vert Y^1,D_1,\xi}$ is the posterior for $Y_2(D_1,d_2)$ according to $\varphi_2(\cdot)$. % \footnote{
Formally, 
$$ F_{Y_2(D_1,d)\vert Y^1,D_1,\xi}(y) = \int \left[ \int 1\{\lambda_2(D_1,d,Y^1,a,e)\leq y\}dF_{\varepsilon_2}(e)\right]dF_{\alpha(d)\vert \varphi_2(Y^1,D_1,\xi)}(a), $$
where $F_{\alpha(d)\vert\varphi_2(Y^1,D_1,\xi)}$ denotes the posterior for $\alpha(d)$ at the start of period $t=2$. % }

In period $t=1$, a rational, forward-looking individual chooses:
\begin{equation}\label{eq:D1_rule-apx}
    D_1 = \argmax_{d_1\in\{0,1\}} \int \left[u_1(d_1,y_1)+ \rho V_2(Y_0,y_1,d_1,\xi)\right] dF_{Y_1(d_1)\vert Y_0,\xi}(y_1) + \eta_1(d_1),
\end{equation}
where $\rho$ is the time discount factor, \(F_{Y_1(d)\vert Y_0,\xi}\) is the prior of $Y_1(d)$ implied by $\varphi_1$, and: 
\[V_2(Y_0,Y_1,D_1,\xi)\equiv\int \max_{d_2\in\{0,1\}}  [\mu_2(d_2,Y_0,Y_1,D_1,\xi)+\eta_2(d_2)]\,dF_{\eta_2(\cdot)}.\]
From (\ref{eq:D2_rule-apx}) and (\ref{eq:D1_rule-apx}), it is clear that, with learning from past outcomes, the optimal treatments chosen by forward-looking individuals take the form of (\ref{eq:treatment-choice}). 

\subsection{Biases without parallel trends and sequential exchangeability} \label{apx:biases}

This subsection derives closed-form expressions that quantify the violations of parallel trends and sequential exchangeability. For convenience, define $\vartheta(z)\equiv\phi(z)/[\Phi(z)(1-\Phi(z))]$, which is $\E(Z\mid Z\ge z)-\E(Z\mid Z<z)$ for $Z\sim N(0,1)$.

\vspace{.25cm}

\begin{lemma}\label{lem:ex3-h}
    In (\ref{eq:ex3-D1}), $\Psi_{\bar\sigma}(m)-\Psi_{\sigma_\eta}(m)>0$ for all $m$, and $h'(m)>1+\rho\gamma-\rho/2$. Hence $h$ is strictly increasing whenever $1+\rho\gamma>\rho/2$.
\end{lemma}

\vspace{.25cm}

\begin{proof}[Proof of Lemma \ref{lem:ex3-h}]
Let $a\equiv m-c$, and $ \psi(z)\equiv z\Phi(z)+\phi(z) $ with $ \psi'(z) = \Phi(z) $. Then
\[
    \Psi_\sigma(m)=\int_{-a/\sigma}^{\infty}(a+\sigma z)\phi(z)\,dz=a\,\Phi(a/\sigma)+\sigma\,\phi(a/\sigma)=\sigma\,\psi(a/\sigma).
\]
Moreover, $\partial_m\Psi_\sigma(m)=\Phi(z)$ and $\partial_\sigma\Psi_\sigma(m)=\psi(z)-z\Phi(z)=\phi(z)>0$ where $ z \equiv a/\sigma $. 
By construction, $\bar\sigma>\sigma_\eta$ because $wv_1>0$. 
Hence $\Psi_{\bar\sigma}(m)-\Psi_{\sigma_\eta}(m)=\int_{\sigma_\eta}^{\bar\sigma}\phi\big((m-c)/\sigma\big)\,d\sigma>0$. Next, by the definition of $h$ in Lemma \ref{lem:ex3-rule}, $h'(m)=1+\rho\gamma+\rho\left[\Phi\big(\tfrac{m-c}{\bar\sigma}\big)-\Phi\big(\tfrac{m-c}{\sigma_\eta}\big)\right]$.
For $m\ge c$, the two terms in the bracket lie in $[1/2,1)$; for $m<c$, the bracket is positive.
So, in both cases the bracket is greater than $-1/2$, and $ h'(m) > 1+\rho\gamma-\rho/2 $. % Thus $h'(m)>0$ whenever $1+\rho\gamma > \rho/2$. 
\hfill
\end{proof}

\vspace{.25cm}

\begin{lemma}\label{lem:gauss-select}
Let $(X,S)$ be jointly normal with $\E(S)=\mu_S$ and $\mathrm{Var}(S)=\sigma_S^2>0$. Then
\[
\E(X\mid S\ge0)-\E(X\mid S<0)=\tfrac{\mathrm{Cov}(X,S)}{\sigma_S}\,\vartheta\!\left(-\tfrac{\mu_S}{\sigma_S}\right).
\]
\end{lemma}

\vspace{.25cm}

\begin{proof}[Proof of Lemma \ref{lem:gauss-select}]
Let $Z\equiv(S-\mu_S)/\sigma_S\sim N(0,1)$ and $z_0\equiv-\mu_S/\sigma_S$, so $\{S\ge0\}=\{Z\ge z_0\}$. By joint normality, $\E(X\mid S)=\E(X)+\frac{\mathrm{Cov}(X,S)}{\sigma_S}Z$. 
Hence the difference of the two conditional means in the lemma is $\frac{\mathrm{Cov}(X,S)}{\sigma_S}\left(\E(Z\mid Z\ge z_0)-\E(Z\mid Z<z_0)\right)$, where $\E(Z\mid Z\ge z_0)-\E(Z\mid Z<z_0)=\frac{\phi(z_0)}{1-\Phi(z_0)}+\frac{\phi(z_0)}{\Phi(z_0)}\equiv \vartheta(z_0)$. \hfill 
\end{proof}

\vspace{.25cm}

Define
\begin{equation}\label{eq:ex3-SEgap}
    \Delta^{SE}\equiv \E[Y_1(1)\mid D_1=1]-\E[Y_1(1)\mid D_1=0].
\end{equation}

\vspace{.25cm}

\begin{proposition}\label{prop:ex3-gaps}
Suppose the conditions of Proposition \ref{prop:ex3-SEPT} hold. 
\begin{enumerate}[(a)]
    \item If $\gamma\neq1$ and $\sigma_{\xi\alpha}=0$, then
        \begin{equation}\label{eq:ex3-PTgap-cf}
            \Delta^{PT}=(\gamma-1)\,\tfrac{b\,\sigma_\varepsilon^2}{\sigma_S}\,\vartheta\!\left(\tfrac{c-a_\alpha}{\sigma_S}\right),
        \end{equation}
        with $b\equiv\sigma_\alpha^2/(\sigma_\alpha^2+\sigma_\varepsilon^2)$ and $\sigma_S^2\equiv b\sigma_\alpha^2+\sigma_\eta^2$.
    \item If $\sigma_{\xi\alpha}\neq0$ and $1+\rho\gamma>\rho/2$, then $\Delta^{SE}>0$. For myopic units $\rho=0$, 
        \begin{equation}\label{eq:ex3-SEgap-cf}
            \Delta^{SE}=\tfrac{\sigma_{\xi\alpha}^2/\sigma_\xi^2}{\sigma_R}\,\vartheta\!\left(\tfrac{c-a_\alpha}{\sigma_R}\right),
        \end{equation}
        where $\sigma_R^2\equiv\sigma_{\xi\alpha}^2/\sigma_\xi^2+\sigma_\eta^2$.
\end{enumerate}
\end{proposition}

\vspace{.25cm}

\begin{proof}[Proof of Proposition \ref{prop:ex3-gaps}]
\emph{Part (a).} If $\sigma_{\xi\alpha}=0$, then $m_1(\xi)=a_\alpha$, $v_1=\sigma_\alpha^2$, $w=b$, and $D_1=\ind\{h(a_\alpha)+\eta_1\ge c\}$ is independent of $(\alpha,\varepsilon_1,\eta_2)$. On $\{D_1=1\}$, $D_2=\ind\{S\ge0\}$ with $S\equiv(1-b)a_\alpha+b(\alpha+\varepsilon_1)-c+\eta_2$, and $(\varepsilon_1,S)$ is jointly normal with $\mu_S=a_\alpha-c$, $\sigma_S^2=b^2(\sigma_\alpha^2+\sigma_\varepsilon^2)+\sigma_\eta^2=b\sigma_\alpha^2+\sigma_\eta^2$ and $\mathrm{Cov}(\varepsilon_1,S)=b\sigma_\varepsilon^2$. Lemma \ref{lem:gauss-select} and (\ref{eq:ex3-pf-trend}) give (\ref{eq:ex3-PTgap-cf}).

\emph{Part (b).} Let $\nu\equiv\alpha-m_1(\xi) = \alpha-a_\alpha-(\sigma_{\xi\alpha}/\sigma_\xi^2)(\xi-a_\xi)$. 
Then $\E(\nu)=0$ and $\mathrm{Cov}(\nu,\xi)=\sigma_{\xi\alpha}-(\sigma_{\xi\alpha}/\sigma_\xi^2)\,\sigma_\xi^2=0$, so by joint normality $\nu\perp\xi$, and $\nu\perp\eta_1$ by (I1). Since $D_1$ is a function of $(\xi,\eta_1)$ and $\varepsilon_1\perp(\xi,\eta_1)$,
\[
\E[Y_1(1)\mid D_1]=\E(\alpha\mid D_1)=\E[m_1(\xi)\mid D_1]+\E(\nu\mid D_1)=\E[m_1(\xi)\mid D_1].
\]
By (\ref{eq:ex3-D1}) and $\eta_1\perp\xi$, $\Pr(D_1=1\mid m_1(\xi)=m_1)=\Phi\big((h(m_1)-c)/\sigma_\eta\big)$, which lies in $(0,1)$ and, by \Cref{lem:ex3-h}, is strictly increasing in $m_1$ when $1+\rho\gamma>\rho/2$. 
With $\sigma_{\xi\alpha}\neq0$, $m_1(\xi)=a_\alpha+(\sigma_{\xi\alpha}/\sigma_\xi^2)(\xi-a_\xi)$ is non-degenerate in $\xi$. \Cref{lem:monotone-select} with $X\equiv m_1(\xi)$ and $A\equiv\{D_1=1\}$ gives $\E(m_1\mid D_1=1)>\E(m_1\mid D_1=0)$, i.e., $\Delta^{SE}>0$.

If $\rho=0$, then $h(m)=m$ and $D_1=\ind\{R\ge0\}$ with $R\equiv a_\alpha+(\sigma_{\xi\alpha}/\sigma_\xi^2)(\xi-a_\xi)-c+\eta_1$. The pair $(\alpha,R)$ is jointly normal with $\mu_R=a_\alpha-c$, $\sigma_R^2=(\sigma_{\xi\alpha}/\sigma_\xi^2)^2\sigma_\xi^2+\sigma_\eta^2=\sigma_{\xi\alpha}^2/\sigma_\xi^2+\sigma_\eta^2$ and $\mathrm{Cov}(\alpha,R)=(\sigma_{\xi\alpha}/\sigma_\xi^2)\,\sigma_{\xi\alpha}=\sigma_{\xi\alpha}^2/\sigma_\xi^2$. Lemma \ref{lem:gauss-select} gives (\ref{eq:ex3-SEgap-cf}). \hfill
\end{proof}

\newpage
\renewcommand*{\bibfont}{\small}
\bibliographystyle{econometrica}
\bibliography{bibliography}

\end{document}